\documentclass[11pt]{article}

\usepackage{amsmath,amssymb}
\usepackage{xcolor}
\usepackage{graphicx}
\usepackage{booktabs}
\usepackage{caption}
\usepackage{hyperref}
\usepackage{enumerate}
\usepackage{enumitem}
\usepackage[noend]{algpseudocode}
\usepackage{lineno}
\usepackage[ruled,vlined]{algorithm2e}

\makeatletter
\def\BState{\State\hskip-\ALG@thistlm}
\makeatother

\newcommand{\REMOVE}[1]{}

\newcommand{\eps}{\varepsilon}
\renewcommand{\phi}{\varphi}

\newcommand{\Part}{\text{Part}}
\newcommand{\Bal}{\textsf{Bal}}

\newcommand{\distr}{district}
\newcommand{\Distr}{District}

\newboolean{lncs}   

\setboolean{lncs}{false}  

\ifthenelse{\boolean{lncs}}{
\pagestyle{headings}
}{
\usepackage{amsthm}
\usepackage{fullpage}

\newtheorem{theorem}{Theorem}
\newtheorem{lemma}[theorem]{Lemma}

\newtheorem{corollary}[theorem]{Corollary}

}

\renewcommand{\emph}{\textbf}

\ifthenelse{\boolean{lncs}}{
\newcounter{section-preserve}
\newcounter{theorem-preserve}
\newcounter{lemma-preserve}
\newcommand{\blank}[1]{}
\newtoks\magicAppendix
\magicAppendix={}
\newtoks\magictoks
\newif\iflater
\laterfalse
\long\def\later#1{\magictoks={#1}%
 \edef\magictodo{\noexpand\magicAppendix={\the\magicAppendix \par
   \the\magictoks%
 }}
 \magictodo}
\long\def\both#1{\magictoks={#1}%
 \edef\magictodo{\noexpand\magicAppendix={\the\magicAppendix \par
   \noexpand\setcounter{theorem-preserve}{\noexpand\arabic{theorem}}%
   \noexpand\setcounter{theorem}{\arabic{theorem}}%
   \noexpand\setcounter{lemma-preserve}{\noexpand\arabic{lemma}}%
   \noexpand\setcounter{lemma}{\arabic{lemma}}%
   \noexpand\setcounter{section-preserve}{\noexpand\arabic{section}}%
   \noexpand\setcounter{section}{\arabic{section}}%
   \noexpand\let\noexpand\oldsection=\noexpand\thesection
   \noexpand\def\noexpand\thesection{\thesection}
   \noexpand\let\noexpand\oldlabel=\noexpand\label
   \noexpand\let\noexpand\label=\noexpand\blank
   \the\magictoks%
   \noexpand\setcounter{theorem}{\noexpand\arabic{theorem-preserve}}%
   \noexpand\setcounter{lemma}{\noexpand\arabic{lemma-preserve}}%
   \noexpand\setcounter{section}{\noexpand\arabic{section-preserve}}%
   \noexpand\let\noexpand\thesection=\noexpand\oldsection
   \noexpand\let\noexpand\label=\noexpand\oldlabel
 }}
 \magictodo
 \the\magictoks}
\def\magicappendix{\latertrue \the\magicAppendix}

\newcommand{\proofsketch}{\vspace*{-1ex} \noindent \textit{Proof Sketch. }}
}{
\newcommand{\both}[1]{#1}
\newcommand{\later}[1]{#1}
}

\ifthenelse{\boolean{lncs}}{
\title{Reconfiguration of Connected Graph Partitions via Recombination\thanks{Supported by NSF CCF-1422311, CCF-1423615, DMS-1800734, and NSERC.}}
}{
\title{Reconfiguration of Connected Graph Partitions via Recombination\thanks{Research supported in part by NFS awards CCF-1422311, CCF-1423615, DMS-1800734, and by NSERC.}}
}

\ifthenelse{\boolean{lncs}}{
\author{
Hugo A. Akitaya\inst{1}
\and
Matias Korman\inst{2}
\and
Oliver Korten\inst{3}
\and\\
Diane L. Souvaine\inst{2}
\and
Csaba D. T\'oth\inst{2,4}
}
\institute{School of Computer Science, Carleton University, Ottawa, ON, Canada. \and
Department of Computer Science, Tufts University, Medford, MA, USA. \and
Department of Computer Science, Columbia University, New York, NY, USA. \and
Department of Mathematics, Cal State Northridge, Los Angeles, CA, USA.}
}{
\author{
Hugo A. Akitaya\thanks{School of Computer Science, Carleton University, Ottawa, ON, Canada.}
\and
Matias Korman\thanks{Department of Computer Science, Tufts University, Medford, MA, USA.}
\and
Oliver Korten\thanks{Department of Computer Science, Columbia University, New York, NY, USA.}
\and
Diane L. Souvaine\footnotemark[3]
\and
Csaba D. T\'oth~\thanks{Department of Mathematics, 
California State University Northridge, Los Angeles, CA, USA.}~\footnotemark[3]
}
}

\begin{document}
	
\maketitle
\ifthenelse{\boolean{lncs}}{
\linenumbers
}{}

\begin{abstract}
Motivated by applications in gerrymandering detection, we study a reconfiguration problem on connected partitions of a connected graph $G$.
A partition of $V(G)$ is \emph{connected} if every part induces a connected subgraph.
In many applications, it is desirable to obtain parts of roughly the same size, possibly with some slack $s$.
A \emph{Balanced Connected $k$-Partition with slack $s$}, denoted \emph{$(k,s)$-BCP}, is a partition of $V(G)$ into $k$ nonempty subsets, of sizes $n_1,\ldots , n_k$ with $|n_i-n/k|\leq s$, each of which induces a connected subgraph (when $s=0$, the $k$ parts are perfectly balanced, and we call it \emph{$k$-BCP} for short).

A \emph{recombination} is an operation that takes a $(k,s)$-BCP of a graph $G$ and produces another by merging two adjacent subgraphs and repartitioning them. Given two $k$-BCPs, $A$ and $B$, of $G$ and a slack $s\geq 0$, we wish to determine whether there exists a sequence of recombinations that transform $A$ into $B$ via $(k,s)$-BCPs. We obtain four results related to this problem:
(1) When $s$ is unbounded, the transformation is always possible using at most $6(k-1)$ recombinations.
(2) If $G$ is Hamiltonian, the transformation is possible using $O(kn)$ recombinations for any $s \ge n/k$, and
(3) we provide negative instances for $s \leq n/(3k)$.
(4) We show that the problem is PSPACE-complete when $k \in O(n^{\eps})$ and $s \in O(n^{1-\eps})$, for any constant $0 < \eps \le 1$, even for restricted settings such as when $G$ is an edge-maximal planar graph or when $k=3$ and $G$ is planar.
\end{abstract}

\section{Introduction}
\label{sec:intro}

Partitioning the vertex set of a graph $G=(V,E)$ into $k$ nonempty subsets $V=\bigcup_{i=1}^kV_i$ that each induces a connected graph $G[V_i]$ is a classical problem, known as the \emph{Connected Graph Partition} problem~\cite{Gyori76,Lovasz77}. Motivated by fault-tolerant network design and facility location problems, it is part of a broader family of problems where each induced graph $G[V_i]$ must have a certain graph property (e.g., $\ell$-connected or $H$-minor-free). In some instances,
it is desirable that the parts $V_1,\ldots , V_k$ have the approximately the same size (depending on some pre-established threshold).
A \emph{Balanced Connected $k$-Partition} (for short, \emph{$k$-BCP}) is a connected partition requiring that $|V_i|=n/k$, for $i \in \{1, \ldots, k\}$ where $n=|V(G)|$ is the total number of vertices. Dyer and Frieze~\cite{DyerF85} proved that finding a $k$-BCP is NP-hard for all $2\leq k\leq n/3$. For $k=2,3$ the problem can be solved efficiently when $G$ is bi- or triconnected, respectively~\cite{SuzukiTN90,WadaK93}, and is equivalent to the perfect matching problem for $k=n/2$. 
Later Chleb{\'{i}}kov{\'{a}}~\cite{Chlebikova96} and Chataigner et al.~\cite{WakabayashiCS07} obtained approximation and inapproximability results for maximizing the ``balance'' ratio $\max_i|V_i|/\min_j|V_j|$ over all connected $k$-partitions. See also~\cite{ItoZN06,LariRPS16,SoltanYZ20,WuZW16} for variants under various other optimization criteria.

In this paper, our basic element is a connected $k$-partition of a graph $G=(V,E)$ that is balanced up to some additive threshold that we call a \emph{slack} $s\geq 0$, denoted \emph{$(k,s)$-BCP}. We explore the space of all \emph{$(k,s)$-BCP}s of the graph $G=(V,E)$.  Note that the total number of \emph{$(k,s)$-BCP}s for all $s\geq 0$, is bounded above by the number $k$-partitions of $V$, which is the Stirling number of the second kind $S(n,k)$, and asymptotically equals $(1+o(1))k^n/k!$ for constant $k$. This bound is the best possible for the complete graph $G=K_n$.

In a recent application~\cite{AGR+19,duchin2018gerrymandering,NDS19}, $G=(V,E)$ represents the adjacency graph of precincts in an electoral map, which should be partitioned into $k$ {\distr s} $V_1,\ldots ,V_k$ where each {\distr} will elect one representative.
Motivated by the design and evaluation of electoral maps under optimization criteria designed to promote electoral fairness, practitioners developed empirical methods to sample the configuration space of potential{\distr} maps by a random walk on the graph where each step corresponds to some elementary \emph{reconfiguration move}~\cite{PrincetonStudy}.
From a theoretical perspective, the stochastic process converges to uniform sampling~\cite{JerrumVV86,LevinPeres17}. However, the move should be \emph{local} to allow efficient computation of each move, and it should support rapid mixing (i.e., the random walk should converge, in total variation distance, to its steady state distribution in time polynomial in $n$).
Crucially, the space of (approximately balanced) $k$-partitions of $G$ must be \emph{connected} under the proposed move. Previous research considered the \emph{single switch} move, in which a single vertex $v\in V$ switches from one set $V_i$ to another set $V_j$ (assuming that both $G[V_i]$ and $G[V_j]$ remain connected).  Akitaya et al.~\cite{akitaya2019reconfiguration} proved that the configuration space is connected under single switch moves if $G$ is biconnected, but in general it is NP-hard both to decide whether the space is connected and to find a shortest path between two valid $k$-partitions. While the single switch is local, both worst-case constructions
and empirical evidence~\cite{DDS19,NDS19} indicate that it does not support rapid mixing. 

In this paper we consider a different move. Specifically, we consider the configuration space of $k$-partitions under the \emph{recombination} move, proposed by DeFord et al.~\cite{DDS19}, in which the vertices in $V_i\cup V_j$, for some $i,j\in \{1,\ldots , k\}$, are re-partitioned into $V_i'\cup V_j'$ such that both $G[V_i']$ and $G[V_j']$ are connected. We also study variants restricted to balanced or near-balanced partitions, that is, when $|V_i|=n/k$ for all $i\in \{1,\ldots ,k\}$, or when $\big| |V_i|-n/k\big|\leq s$ for a given slack $s\geq 0$. In application domains mentioned above, the underlying graph $G$ is often planar or near-planar, and in some cases it is a triangulation (i.e., an edge-maximal planar graph). Results pertaining to these special cases are of particular interest.

\paragraph{Definitions.}
Let $G=(V,E)$ be a graph with $n=|V(G)|$. For a positive integer $k$,
a \emph{connected $k$-partition} $\Pi$ of $G$ is a partition of $V(G)$ into disjoint nonempty subsets $\{V_1,\ldots,V_k\}$ such that
the induced subgraph $G[V_i]$ is connected for all $i\in\{1,\ldots,k\}$.
Each subgraph induced by $V_i$ is called a \emph{{\distr}}.
We write $\Pi(v)$ for the subset in $\Pi$ that contains vertex $v$.

Denote by $\Part(G,k)$ the set of connected $k$-partitions on $G$. We also consider subsets of $\Part(G,k)$ in which all {\distr s} have the same or almost the same number of vertices. A connected $k$-partition of $G$ is \emph{balanced} ($k$-BCP) if every {\distr} has precisely $n/k$ vertices (which implies that $n$ is a multiple of $k$); and it is  \emph{balanced with slack $s\geq 0$} ($(k,s)$-BCP), if $\big| |U|-n/k\big|\leq s$ for every {\distr} $U\subset V$. Let $\Bal_s(G,k)$ denote the set of connected $k$-partitions on $G$ that are balanced with slack $s$, i.e., the set of all $(k,s)$-BCPs. The set of balanced $k$-partitions is denoted $\Bal(G,k)=\Bal_0(G,k)$; and $\Part(G,k)=\Bal_\infty(G,k)$.

We now formally define a \emph{recombination move} as a binary relation on $\Bal_s(G,k)$.
\sloppy Two non-identical $(k,s)$-BCPs, $\Pi_1 = \{V_1, \ldots, V_k\}$ and $\Pi_2 = \{W_1, \ldots, W_k\}$,are related by a recombination move if there exist $i,j\in \{1,\ldots , k\}$,
and a permutation $\pi$ on $\{1,\ldots , k\}$ such that $V_i\cup V_j=W_{\pi(i)}\cup W_{\pi(j)}$
and $V_\ell=W_{\pi(\ell)}$ for all $\ell\in \{1,\ldots , k\}\setminus \{i,j\}$.
We say that $\Pi_1$ and $\Pi_2$ are a recombination of each other.
This binary relation is symmetric and defines a graph on $\Bal_s(G,k)$ for all $s\geq 0$.
This graph is the \emph{configuration space} of $\Bal_s(G,k)$ under recombination,
denoted by $\mathcal{R}_s(G,k)$.

\paragraph{Balanced Recombination Problem  BR$(G,k,s)$:}
Given a graph $G=(V,E)$ with $|V|=n$ vertices and two $(k,s)$-BCPs $A$ and $B$,
decide whether there exists a path between $A$ and $B$ in $\mathcal{R}_s(G,k)$, i.e. whether there is a sequence of recombination moves that carries $A$ to $B$ such that every
intermediate partition is a $(k,s)$-BCP.

\paragraph{Our Results.}
We prove, in Section~\ref{sec:connected}, that the configuration space $\mathcal{R}_\infty(G,k)$ is connected whenever the underlying  graph $G$ is connected and the size of the {\distr s} is unrestricted. It is easy, however, to construct a graph $G$ where $\mathcal{R}_0(G,k)$ is disconnected. We study what is the minimum slack $s$, as a function of $n$ and $k$, that guarantees that $\mathcal{R}_s(G,k)$ is connected for all connected (or possibly biconnected) graphs $G$ with $n$ vertices. We prove that $\mathcal{R}_s(G,k)$ is connected and its diameter is $O(nk)$ for $s=n/k$ when $G$ is a Hamiltonian graph (Section~\ref{sec:hamiltonian}). As a counterpart, we construct a family of Hamiltonian planar graphs $G$ such that $\mathcal{R}_s(G,k)$ is disconnected for $s<n/(3k)$ (Section~\ref{sec:disconnected}).

We prove in Section~\ref{sec:hardness} that BR$(G,k,s)$ is PSPACE-complete even for the special case when $G$ is a triangulation (i.e., an edge-maximal planar graph), $k$ is $O(n^{\varepsilon})$  and $s$ is $O(n^{1-\varepsilon})$ for constant $0 < \varepsilon \le 1$. As a consequence we show that finding a $(k,s)$-BCP of $G$ is NP-hard in the same setting. 
Note that the previously known hardness proofs for finding $k$-BCPs either require that $G$ is weighted and nonplanar~\cite{WakabayashiCS07} or $G$ contain cut vertices~\cite{DyerF85}.
In contrast, if $G$ is planar and 4-connected, then $G$ admits a Hamilton cycle~\cite{tutte1956theorem} and, therefore, a $(k,s)$-BCP is easily obtained by partitioning a Hamilton cycle into the desired pieces.
Finally, we modify our construction to also show that BR$(G,k,s)$ is PSPACE-complete even for the special case when $G$ is planar, $k = 3$, and $s$ is bounded above by $O(n^{1-\varepsilon})$ for constant $0 < \varepsilon \le 1$.

\section{Recombination with Unbounded Slack}
\label{sec:connected}

In this section, we show that the configuration space $\mathcal{R}_\infty(G,k)$ is connected under recombination moves if $G$ is connected (cf.~Theorem~\ref{thm:connected}). The proof proceeds by induction on $k$, where the induction step depends on Lemma~\ref{lem:induct} below.

We briefly review some standard graph terminology. A \emph{block} of a graph $G$ is a maximal biconnected component of $G$.
A vertex $v\in V(G)$ is a \emph{cut vertex} if it lies in two or more blocks of $G$, otherwise it is a \emph{block vertex}.
In particular, if $v$ is a block vertex, then $G-v$ is connected.
If $G$ is a connected graph with two or more vertices, then every block has at least two vertices.
A block is a \emph{leaf-block} if it contains precisely one cut vertex of $G$.
Every connected graph either is biconnected or has at least two leaf blocks.
The \emph{arboricity} of a graph $G$ is the minimum number of forests that cover all edges in $E=(G)$.
The \emph{degeneracy} of $G$ is the largest minimum vertex degree over all induced subgraphs of $G$.
It is well known that if the arboricity of a graph is $a$, then its degeneracy is between $a$ and $2a-1$.

\ifthenelse{\boolean{lncs}}{\later{
\section{Omitted proofs from Sections~\ref{sec:connected}, \ref{sec:hamiltonian} and \ref{sec:disconnected}}\label{sec:omitted}
}}{}

\both{
\begin{lemma}\label{lem:arbor}
If the arboricity of a graph is $a$, then it contains a block vertex of degree at most $2a-1$.
\end{lemma}
}

\later{
\begin{proof}
It is enough to prove the claim for a connected component of $G$, so we may assume that $G$ is connected.
First assume that $G$ is biconnected. Then every vertex is a block vertex. Since the degeneracy of $G$ is at most $2a-1$, there exists a vertex of degree at most $2a-1$, as required.
Next assume that $G$ is not biconnected. Let $G[U]$ be a leaf-block of $G$, and let $u\in U$ be the unique cut-vertex of $G$ in $U$. Since the degeneracy of $G$ is at most $2a-1$, there exists a vertex in $U$ whose degree in $G[U]$ is at most $2a-1$.
If a block vertex $v\in U\setminus \{u\}$ has degree at most $2a-1$, our proof is complete.
Suppose, to the contrary, that the degree of every block vertex in $U$ is at least $2a$ (and the degree of the unique cut vertex in $G[U]$ is at least 1).
By the handshake lemma, the number of edges in $G[U]$ is at least $\frac12(2a(|U|-1)+1)>a(|U|-1)$. However, the arboricity of $G[U]$ is at most $a$, and $a$ forests on the vertex set $U$ jointly contain at most $a(|U|-1)$ edges, which provides a contradiction.\ifthenelse{\boolean{lncs}}{\hfill$\Box$}{}
\end{proof}
}

\ifthenelse{\boolean{lncs}}{The proof of Lemma~\ref{lem:arbor} can be found in Appendix~\ref{sec:omitted}. }{}The heart of the induction step of our main result hinges on the following lemma.

\begin{lemma}\label{lem:induct}
Let $G$ be a connected graph, $k\geq 2$ an integer, and $\Pi_1,\Pi_2\in \Part(G,k)$ be two $k$-partitions of $G$.
Then there exists a block vertex $v \in V(G)$ such that up to three recombination moves can transform $\Pi_1$
and $\Pi_2$ each to two new $k$-partitions in which $\{v\}$ is a singleton distinct.
\end{lemma}
\begin{proof}
Let $\Pi_1=\{V_1,\ldots , V_k\}$ and $\Pi_2=\{W_1,\ldots , W_k\}$. We construct two spanning trees, $T_1$ and $T_2$, for $G$ that each contain $k-1$ edges between the {\distr s} of $\Pi_1$ and $\Pi_2$, respectively. Specifically, for $i\in \{1,\ldots ,k\}$, let $T(V_i)$ be a spanning tree of $G[V_i]$, $T(W_i)$ a spanning tree of $G[W_i]$. As $G$ is connected, we can augment the forest $\bigcup_{i=1}^k T(V_i)$ to a spanning tree $T_1$ of $G$, using $k-1$ new edges, which connect vertices in distinct {\distr s}. Similarly, we can augment $\bigcup_{i=1}^k T(W_i)$ to a spanning tree $T_2$ of $G$. Now, let $G' = T_1 \cup T_2$.
By definition, the arboricity of $G'$ is at most 2. By Lemma~\ref{lem:arbor}, $G'$ contains a block vertex $v$ with $\deg_{G'}(v)\leq 3$.

We show that we can modify $\Pi_1$ (resp., $\Pi_2$) to create a singleton {\distr} $\{v\}$ in at most three moves. Assume without loss of generality that $v\in V_1$ and $v\in W_1$. Since $\deg_{G'}(v)\leq 3$, we have $\deg_{T(V_1)}(v)\leq 3$ and $\deg_{T(W_1)}(v)\leq 3$. Consequently, $T(v_1)-v$ (resp., $T(W_1)-v$) has at most three components, each of which is adjacent to some other {\distr}, since $G'-v$ is connected. Up to three successive recombinations can decrease the {\distr} $V_1$ with the components of $T(V_1)-v$, and reduce $V_1$ to $\{v\}$. Similarly, at most three successive recombinations can reduce $W_1$ to $\{v\}$. \ifthenelse{\boolean{lncs}}{\hfill$\Box$}{}
\end{proof}

\begin{theorem}\label{thm:connected}
Let $G$ be a connected graph and $k\geq 1$ a positive integer. For all $\Pi_1,\Pi_2\in \Part(G,k)$, there exists a sequence of at most $6(k-1)$ recombination moves that transforms $\Pi_1$ to $\Pi_2$.
\end{theorem}
\begin{proof}
We proceed by induction on $k$. In the base case, $k=1$, and $\Pi_1=\Pi_2$. Assume that $k>1$ and claim holds for $k-1$.
By Lemma~\ref{lem:induct}, we can find a block vertex $v\in V(G)$ and up to six recombination moves transform $\Pi_1$ and $\Pi_2$ into $\Pi_1'$ and $\Pi_2'$ such that both contain $\{v\}$ as a singleton {\distr}. Since $v$ is a block vertex, $G-v$ is connected;
and since $\{v\}$ is a singleton {\distr} in both $\Pi_1'$ and $\Pi_2'$, we have $\Pi_1-\{v\},\Pi_2-\{v\}\in \Part(G-v,k-1)$.
By induction, a sequence of up to $6(k-2)$ recombination moves in $G-v$ can transform $\Pi_1-\{v\}$ $\Pi_2-\{v\}$. These moves remain  valid recombination moves in $G$ if we add singleton {\distr} $\{v\}$. Overall, the combination of these sequences yields a sequence of up to $6+6(k-2)=6(k-1)$ recombination moves that transforms $\Pi_1$ to $\Pi_2$. This completes the induction step.\ifthenelse{\boolean{lncs}}{\hfill$\Box$}{}
\end{proof}

\section{Recombination with Slack}
\label{sec:hamiltonian}

In this section, we prove that the configuration space $\mathcal{R}_s(G,k)$ is connected if the slack is greater or equal to the average {\distr} size, that is, $s\geq n/k$, and the underlying graph $G$ is Hamiltonian (Theorem~\ref{thm:hamiltonian}). 

Let $G$ be a graph with $n$ vertices that contains a Hamilton cycle $C$.
Assume that $n$ is a multiple of $k$. A $k$-partition in $\Bal_s(G,k)$ is \emph{canonical} if each {\distr} consists of consecutive vertices along $C$. Using a slack of $s\geq n/k$, we can transform any canonical $k$-partition to any other using $O(k^2)$ reconfigurations.

\both{
\begin{lemma}\label{lem:canonical}
Let $G$ be a graph with $n$ vertices and a Hamilton cycle $C$, $k\geq 1$ is a divisor of $n$, and $s\geq n/k$. 
Then the subgraph of $\mathcal{R}_s(G,k)$ induced by canonical $k$-partitions is connected and its diameter is at most $k^2+1$.
\end{lemma}
}

\later{
\begin{proof}
Let $\Pi = \{V_1, \ldots, V_k\}$ be a canonical $k$-partition with slack $s$, where {\distr s} $V_i$, $i\in \{1,\ldots,k\}$ are labeled in cyclic order along $C$. 
We call a {\distr} $V_i$ \emph{large} if  $|V_i| > n/k$.
We first show that at most $\binom{k}{2}$ moves are enough to bring any canonical $k$-partition into a balanced canonical partition, i.e., none of its {\distr s} is large. 
Then we show that any pair of balanced canonical $k$-partitions is within $k+1$ moves apart. 
Overall the diameter of $\mathcal{R}_s(G,k)$ is at most $2\binom{k}{2}+k+1=k^2+1$.

Our first proof is by induction on $k-m$, where the first $m$ blocks are balanced
(i.e, $|V_i|=n/k$ for all $i\leq m$). In the base case, $k-m=0$, and all {\distr s} are balanced. In the induction step, assume that $0\leq m<k$. 
Since average size of the {\distr s} $V_{m+1},\ldots ,V_k$ is $n/k$, 
but $|V_{m+1}|\neq n/k$, then there exist {\distr s} both below and above the average.
Let $j>m$ be the first index such that $V_j\leq n/k$ and $V_{j+1}\geq n/k$ or vice versa.
We recombine $V_j$ and $V_{j+1}$ into $V_j'$ and $V_{j+1}'$ such that $|V_j'| = n/k$.
By assumption, $2n/k - s \le |V_j|+|V_{j+1}| \le 2n/k + s$, thus $n/k-s\leq |V_{j+1}'|\leq n/k+s$, hence the new $k$-partition is a $(k,s)$-BCP. 
We then successively recombine $V_i$ and $V_{i-1}$ for $i=j,j-1,\ldots, m+1$. In each of these  recombinations, we set $|V_{i-1}'|=n/k$ and $|V_i'| =|V_{i-1}\cup V_i|-n/k$, thus $n/k-s\leq  |V_i'|\leq n/k+s$. After $j-m-1\leq k-m-1$ moves, we obtain a $(k,s)$-BCP $\Pi''$ with $|V_{m+1}''|=n/k$, and we can increment $m$. 
The number of recombination moves is bounded by $\sum_{m=0}^{k-2}(k-m-1)=\sum_{i=1}^{k-1}i=\binom{k}{2}$.

It remains to show that any two balanced canonical partitions in $\mathcal{R}_s(G,k)$ 
are within $k$ moves apart. Note that if $\Pi_1=\{V_1,\ldots , V_k\}$ and $\Pi_2=\{W_1,\ldots , W_k\}$ are two balanced canonical partitions, then $\Pi_2$ is a cyclic shift of $\Pi_1$ along $C$. Without loss of generality, assume that $V_1 \cap W_1 \neq \emptyset$ and $W_1\subset V_1 \cup V_2$.
We describe $k+1$ moves that brings $\Pi_1$ to $\Pi_2$. This process is similar to the induction step of the previous argument. 
First, recombine $V_1$ and $V_2$ making $V_1' = V_1 - W_2$ and $V_2' = V_2 \cup W_2$.
Then, recombine $V_i'$ and $V_{i+1}$ making making $V_i'' = W_i$ and $V_{i+1}' = V_{i+1} \cup W_{i+1}$ for $i=2,\ldots ,k$.
Finally, we recombine $V_1'$ and $V_k'$ into $W_1$ and $W_k$.
Through this procedure, there are only two unbalanced {\distr s} at any given time; 
the only large {\distr} is $V_1'$ and $|V_1'| < 2n/k$.\ifthenelse{\boolean{lncs}}{\hfill$\Box$}{}
\end{proof}
}

\ifthenelse{\boolean{lncs}}{
\proofsketch{
We proceed by induction: We assume that the first $\ell\in \{0,\ldots , k\}$ 
{\distr s} each have size $\frac{n}{k}$, and we change the size of the $(\ell+1)$st {\distr} to $\frac{n}{k}$ using at most $k-\ell-1$ recombinations. Since the average size of the remaining $k-\ell$ {\distr s}  
is $n/k$, there are two consecutive {\distr s} of size at most $\frac{n}{k}$ and at least $\frac{n}{k}$, respectively. We recombine the first such pair of {\distr s}, and propagate the changes to 
the $(\ell+1)$st {\distr}, completing the induction step. \hfill$\Box$\smallskip
}
}{}

In the remainder of this section, we show that every $k$-partition in $\Bal_s(G,k)$
can be brought into canonical form by a sequence of $O(nk)$ recombinations.

\paragraph{Preliminaries.} We introduce some terminology.
Let $\Pi=\{V_1,\ldots ,V_k\}\in \Bal_s(G,k)$ with a slack of $s\geq n/k$.
For every $i\in \{1,\ldots , k\}$, a \emph{fragment} of $G[V_i]$ is
a maximum set $F\subset V_i$ of vertices that are contiguous along $C$.
Every set $V_i$ is the disjoint union of one or more fragments.
The $k$-partition $\Pi$ is canonical if and only if every {\distr} has precisely one fragment.
Our strategy is to ``defragment'' $\Pi$ if it is not canonical, that is, we
reduce the number of fragments using recombination moves.

We distinguish between two types of {\distr s} in $\Pi$: 
A {\distr} $V_i$ is \emph{small} if $|V_i|\leq n/k$, otherwise it is \emph{large}.
Every edge in $E(G)$ is either an edge or a \emph{chord} of the cycle $C$.
For every $i\in \{1,\ldots , k\}$, let $f_i$ be the number of fragments of $V_i$.
Let $T_i$ be a spanning tree of $G[V_i]$ that contains the minimum number of chords.
The edges of $G[V_i]$ along $C$ form a forest of $f_i$ paths; we can construct $T_i$ by
augmenting this forest to a spanning tree of $G[V_i]$ using $f_i-1$ chords.

The \emph{center} of a tree $T$ is a vertex $v\in V(T)$ such that each component of
$T-v$ has up to $|V(T)|/2$ vertices. It is well known that every tree has a center.
For $i\in \{1,\ldots ,k\}$, let $c_i$ be a center of the spanning tree $T_i$ of $G[V_i]$.
Let the fragment of $V_i$ be \emph{heavy} if it contains $c_i$; and \emph{light} otherwise.
We also define a parent-child relation between the fragments of $V_i$.
Fragments $A$ and $B$ are in a parent-child relation if they are adjacent in $T_i$
and if $c_i$ is closer to $A$ than to $B$ in $T_i$.
Note that a light fragment and its descendants jointly contain less than
$|V_i|/2\leq (n/k+s)/2$ vertices; see Fig.~\ref{fig:fragments}.

	\begin{figure}[htbp]
		\centering
		\includegraphics{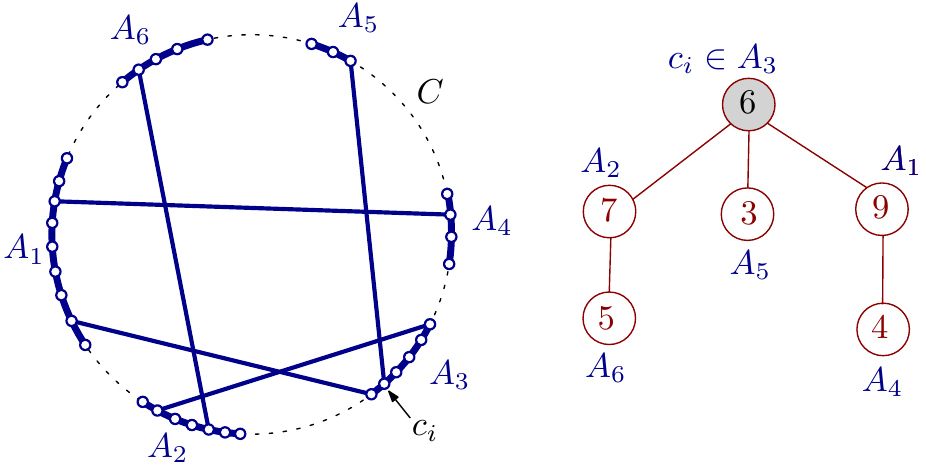}
		\caption{Left: A distinct $V_i$ with 26 vertices (hollow dots) in six fragments (bold arcs) along $C$. 
    The spanning tree $T_i$ of $G[V_i]$ contains five edges, with a center at $c_i$. 
    Right: Parent-child relationship between fragments is defined by the tree rooted at the fragment containing $c_i$.}
		\label{fig:fragments}
	\end{figure}

The following four lemmas show that we can decrease the number of fragments under some conditions.
In all four lemmas, we assume that $G$ is a graph with a Hamiltonian cycle $C$,
and $\Pi$ is a noncanonical $(k,s)$-BCP with $s\geq n/k$.

\begin{lemma}\label{lem:light}
If a light fragment of a large {\distr} is adjacent to a small {\distr} along $C$,
then a recombination move can decrease the number of fragments.
\end{lemma}
\begin{proof}
Assume without loss of generality that $v_1v_2$ is an edge of $C$, where $v_1\in F_1\subset V_1$, $v_2\in F_2\subseteq V_2$, $F_1$ is a light fragment of a large {\distr} $V_1$, and $F_2$ is some fragment of a small {\distr} $V_2$. Let $\overline{F}_1$ be the union of fragment $F_1$ and all its descendants. By the definition of the center $c_1$, we have $|\overline{F}_1|<|V_1|/2$. Apply a recombination replacing $V_1$ and $V_2$ with $W_1=V_1\setminus \overline{F}_1$ and $W_2=V_2\cup \overline{F}_1$.

We show that the resulting partition is a $(k,s)$-BCP. Note also that both $G[\overline{F}_1]$ and $G[V_1\setminus \overline{F}_1]=G[W_1]$ are connected. Since $v_1v_2\in E(G)$, then $G[V_2\cup \overline{F}_1]=G[W_2]$ is also connected.
As $W_1$ contains the center of $V_1$, we have $|W_1|\geq 1$ and $|W_1|<|V_1|\leq n/k+s$. As $V_2$ is small, have $|W_2|=|V_2|+|\overline{F}_1|< n/k +n/k\leq 2n/k\leq n/k+s$. Finally, note that $F_1\cup F_2$ is a single fragment in the resulting $k$-partition, hence the number of fragments decreased by at least 1.\ifthenelse{\boolean{lncs}}{\hfill$\Box$}{}
\end{proof}

\begin{lemma}\label{lem:else}
If a light fragment of a large {\distr} is adjacent to some small {\distr} along $C$,
then there exists two adjacent {\distr s} along $C$ whose combined size is at most $2n/k$.
\end{lemma}
\begin{proof}
Suppose, to the contrary, that every small {\distr} is adjacent only to heavy fragments along $C$, and the combined size of every pair of adjacent {\distr s} along $C$ is greater than $\frac{2n}{k}$, meaning that at least one {\distr} is large. 
We assign every small {\distr} to an adjacent large {\distr} as follows. For every small {\distr} $V_i$, let $F_i$ be one of its arbitrary fragments. 
We assign $V_i$ to the large {\distr} whose heavy fragment is adjacent to $F_i$ in the clockwise direction along $C$. Since every large {\distr} has a unique heavy fragment, and at most one {\distr}  precedes it in clockwise order along $C$, the assignment is a matching of the small {\distr s} to large {\distr s}. 
Denote this matching by $M$. Every {\distr} that is not part of a pair in $M$ must be large. By assumption, every pair in $M$ has combined size greater than $\frac{2n}{k}$, so the average {\distr} size over the {\distr s} in $M$ is greater than $\frac{n}{k}$. The {\distr s} not in $M$ are large so their average size also exceeds $\frac{n}{k}$. Overall the average {\distr} size exceeds $\frac{n}{k}$. But $\Pi$ is a $k$-partition of $n$ vertices, hence the average {\distr} size is exactly $\frac{n}{k}$, a contradiction.\ifthenelse{\boolean{lncs}}{\hfill$\Box$}{}
\end{proof}

\begin{lemma}\label{lem:average}
If {\distr s} $V_1$ and $V_2$ are adjacent along $C$ and $|V_1\cup V_2|\leq n/k+s$, then there is a recombination move that
either decreases the number of fragments, or maintains the same number of fragments and creates a singleton {\distr}.
\end{lemma}
\begin{proof}
Assume, w.l.o.g., that $v_1\in F_1\subseteq V_1$, $v_2\in F_2\subseteq V_2$, where $v_1v_2$ is an edge of $C$, and $F_1$ and $F_2$ are fragments of $V_1$ and $V_2$, respectively.
The induced graph $G[V_1\cup V_2]$ is connected, and $T_1\cup T_2\cup v_1v_2$ is one of its spanning trees.
If $T_1$ or $T_2$ contains a chord, say $e$, then $(T_1\cup T_2\cup v_1v_2)-e$ has two components, $T_3$ and $T_4$, each of size at most $n/k+s-1$. A recombination move can replace $V_1$ and $V_2$ with $V(T_3)$ and $V(T_4)$.
Since fragments $F_1$ and $F_2$ merge into one fragment, the number of fragments decreases by at least one.
Otherwise, neither $T_1$ nor $T_2$ contains a chord. Then $V_1$ and $V_2$ each has a single fragment,
so $V_1\cup V_2$ is a chain of vertices along $C$. Let $v$ be the first vertex in this chain.
A recombination move can replace $V_1$ and $V_2$ with $W_1=\{v\}$ and $W_2=(V_1\cup V_2)\setminus \{v\}$.
By construction both $G[W_1]$ and $G[W_2]$ are connected, $|W_1|=1$, $|W_2|=|V_1\cup V_2|-1\leq n/k+s-1$,
and the number of fragments does not change.\ifthenelse{\boolean{lncs}}{\hfill$\Box$}{}
\end{proof}

\begin{lemma}\label{lem:singleton}
If there exists a singleton {\distr}, then there exists a sequence of at most $k-1$ recombination moves that decreases the number of fragments.
\end{lemma}
\begin{proof}
Let $C=(v_1,\ldots , v_n)$. Assume without loss of generality that $V_1=\{v_1\}$ is a singleton {\distr}, and $v_2\in F_2\subseteq V_2$, where $F_2$ is a fragment of {\distr} $V_2$. Since not all {\distr s} are singletons, we may further assume that $|V_2|\geq 2$.
We distinguish between two cases.

Case~1: $F_2\neq V_2$ (i.e., $V_2$ has two or more fragments). Let $e$ be an arbitrary chord in $T_2$, and denote the two subtrees of $T_2-e$ by $T_2^-$ and $T_2^+$ such that $v_2$ is $T_2^-$. Since $|V_2|\leq n/k+s$, the subtrees $T_2^-$ and $T_2^+$ each have at most $n/k+s-1$ vertices. We can recombine $V_1$ and $V_2$ into $W_1=V_1\cup V(T_2^-)$ and $W_2=V(T_2^+)$. Then $|W_1|\leq 1+(n/k+s-1)=n/k+s$ and $|W_2|\leq n/k+s-1$; they both induce a connected subgraph of $G$. As the singleton fragment $V_1$ and $F_2$ merge into one fragment of $W_1$, the number of fragments decreases by at least one.

Case~2: $F_2=V_2$ (i.e., {\distr} $V_2$ has only one fragment). Let $t>2$ be the smallest index such that $v_t$ is in a {\distr} that has two or more fragments (such {\distr} exists since $\Pi$ is not canonical). Then the chain $(v_1,\ldots , v_{t-1})$ is covered by single-fragment {\distr s} that we denote by $V_1,\ldots ,V_\ell$ along $C$. By recombining $V_i$ and $V_{i+1}$ for $i=1,\ldots ,\ell-1$, we create new single-fragment {\distr s} $W_1,\ldots ,W_\ell$ such that $|W_i|=|V_{i+1}|$ for $i=1,\ldots , \ell+1$ and $|W_\ell|=|V_1|=1$. Now we can apply Case~1 for the singleton {\distr} $W_\ell$.\ifthenelse{\boolean{lncs}}{\hfill$\Box$}{}
\end{proof}

We are now ready to prove the main result of this section.

\begin{theorem}\label{thm:hamiltonian}
If $G$ is a Hamiltonian graph on $n$ vertices and $s\geq n/k$,
then $\mathcal{R}_s(G,k)$ is connected and its diameter is $O(nk)$.
\end{theorem}
\begin{proof}
Based on Lemmas~\ref{lem:light}--\ref{lem:singleton}, the following algorithm successively reduces the number of fragments to $k$, thereby transforming any balanced $k$-partition to a canonical partition.\\
While the number of fragments is more than $k$, do:
\begin{enumerate}
\item If a fragment of a small {\distr} is adjacent to a light fragment of a large {\distr} along $C$, then apply the recombination move in Lemma~\ref{lem:light}, which decreases the number of fragments.
\item Else, by Lemma~\ref{lem:else}, there are two adjacent {\distr s} along $C$ whose combined size is at most $2n/k$. Apply a recombination move in Lemma~\ref{lem:average}. If this move does not decrease the number of fragments, it creates a singleton {\distr}, and then up to $k-1$ recombination moves in Lemma~\ref{lem:singleton} decrease the number of fragments by at least one.
\end{enumerate}
There can be at most $n$ different fragments in a $k$-partition of a set of $n$ vertices.
We can reduce the number of fragments using up to $k$ recombination moves. Overall, $O(nk)$ recombination moves can bring any two $(k,s)$-BCPs to canonical form, which are within $k^2+1$ moves apart by Lemma~\ref{lem:canonical}.\ifthenelse{\boolean{lncs}}{\hfill$\Box$}{}
\end{proof}

\section{Disconnected Configuration Space}
\label{sec:disconnected}

In this section we show that the configuration space is not always connected, even in Hamiltonian graphs.
Specifically, we show the following result:

\both{
\begin{theorem}
\label{thm:neg-example}
	For any $k\geq 4$ and $s>0$ there exists a Hamiltonian planar graph $G$ of $n=k(3s+2)$ vertices such that $\mathcal{R}_s(G,k)$ is disconnected.
\end{theorem}
}

	\begin{figure}[htbp]
		\centering
		\includegraphics[width=0.6\linewidth]{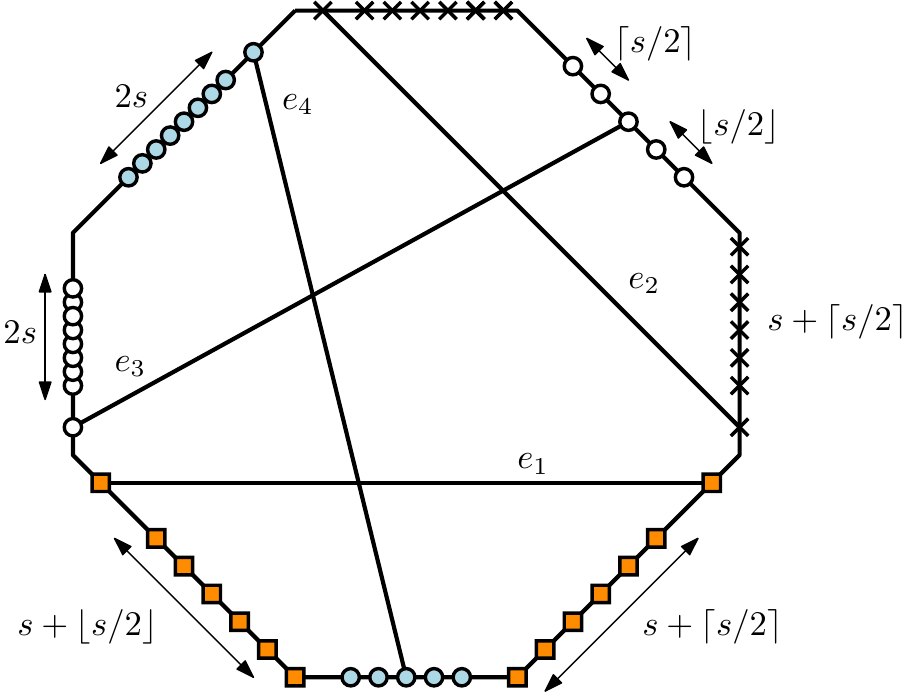}
		\caption{Problem instance showing that $\mathcal{R}_s(G,k)$ is not always connected (for $k=4$, $n=56$ and $s=4 = \frac{n}{3k}-O(1)$).}
		\label{fig:disconnected}
	\end{figure}

\later{
\begin{proof}	
	For simplicity, we first argue about the construction for $k=s=4$ and thus $n=56$ (an extension to larger values of $k$ is given afterwards). Our instance consists of a cycle (shown as an octagon in Fig.~\ref{fig:disconnected}) and four additional edges $e_1, \ldots ,e_4$ called \emph{chords}. 
	Note that $G$ is planar despite the non-planar drawing in Fig.~\ref{fig:disconnected} (we may draw $e_1$ and $e_3$ in the outer face).
All four {\distr s} initially have $n/4=14$ vertices.
	
	Note that, if we were to remove any of the chords, the corresponding {\distr} would be split into two connected components. The removal of the chords $e_1,\ldots , e_4$ would split two of the {\distr s} (marked with circles in Fig.~\ref{fig:disconnected})
into two components of sizes $1+s$ and $1+2s$, resp., and it would split the other two {\distr s} into two components of the same size within a vertex difference.

When a chord $e$ is critical for the connectivity of the induced graph $G[V_i]$ of a {\distr} $V_i$, we say that $V_i$ is \emph{split} (by $e$). 	
As noted above, all four {\distr s} are split in the initial partition. We claim that no sequence of recombinations can change this fact. Assume, for the sake of contradiction, that after a sequence of recombinations, one of the {\distr s} $D$ is not split. Consider the last recombination before this happens.

This recombination involves two {\distr s} $V_i$ and $V_j$ that are split via edges $e_i$ and $e_j$, respectively. 
The union of the two {\distr s} would have three components (after removing $e_i$ and $e_j$). 
Thus, at least one of the two {\distr s} must remain split after the recombination (that is, we cannot ``unsplit" both {\distr s} in one recombination).

Next, we show that not even one {\distr} can become unsplit. We distinguish between four cases depending on the edge that initially split a {\distr}, but it no longer does so after the recombination:
	
\begin{itemize}
		\item We first argue that this edge can be neither $e_1$ nor $e_2$. We first argue for $e_1$. By the assumption that $D$ is the first unsplit {\distr}, the other chords ($e_2$, $e_3$, and $e_4$) each split other {\distr s}. In particular, $D$ is contained in a contiguous arcs of $C$ that does not include an endpoint of any other chord. Since $e_1$ was involved in the recombination, we may assume that $D$ is contained in the lower right or lower left portion of Fig~\ref{fig:disconnected}, between the edges of $e_3$ and $e_4$ or $e_2$ and $e_4$. In either case, $D$ can contain at most $2s+1$ vertices. Since $n/k = 3s+2$, this is not possible unless the slack was at least $s+1$. Edge $e_2$ can be ruled out with an symmetric argument.
		
		\item Edge $e_3$ is ruled out using a similar argument. 
		Assume again that $D$ is the first unsplit {\distr}, and the other three chords each split other {\distr s}. Then $D$ is contained in a contiguous arc of $C$ between  endpoints of $e_1$, $e_2$, or $e_4$, and this arc contains an endpoint of $e_3$.
		Assume first that $D$ is in an arc containing the left endpoint of $e_3$. Then 
		the right endpoint of $e_3$ is in a {\distr} $D'$ that contains both endpoints of $e_2$. However, $D'$ contains an arc of $C$ between the endpoints of $e_2$, 
		and $e_2$ does not split any {\distr}, a contradiction.
		Next assume that $D$ is in the arc containing the right endpoint of $e_3$. 
		Then $D$ is contained in an arc between two endpoints of $e_2$, and so $D$ is adjacent to the {\distr} split by $e_2$. The union of these two {\distr s} lies in a continuous arc of $4s+3$ vertices. However, the minimum size of two {\distr s} is $2(n/k-s)=2(2s+2)$; a contradiction. 
		By symmetry, we also rule out $e_4$, and thus we obtain that all {\distr s} must be split at all times.
	\end{itemize}
	\ifthenelse{\boolean{lncs}}{}{\newpage}
	
Finally, it remains to show how to extend the proof for larger values of $k$. 
Observe that in the proof we never looked at the upper left side of the octagon (where we have a cluster of $2s$ vertices of one {\distr} followed by $2s$ vertices of another {\distr}). 
For each additional {\distr} we need, we can simply place $2+3s$ consecutive vertices. Those {\distr s} can do local recombinations, but they will not prevent the initial four {\distr s} from being split.\ifthenelse{\boolean{lncs}}{\hfill$\Box$}{}
\end{proof}
}
\ifthenelse{\boolean{lncs}}{
\proofsketch{
The proof is constructive and can be found in Appendix~\ref{sec:omitted}. We construct an instance that consists of a cycle and $4$ chords (shown in Fig.~\ref{fig:disconnected}). Each {\distr} consists of two contiguous arcs along the cycle, which are connected by a unique chord. We prove that no sequence of recombinations can change this fact. 
Indeed, the chords are sufficiently limiting that a {\distr} can only gain/lose vertices in a very restricted fashion (e.g., the {\distr} of represented by orange squares can gain up to $s$ vertices from the {\distr} of blue circles). \hfill$\Box$
}
}{}

\section{Hardness Results}
\label{sec:hardness}

This section presents our hardness results.
Our reductions are from Nondeterministic Constrained Logic (NCL) reconfiguration which is PSPACE-complete~\cite{PSPACE,PSPACE-book}.
An instance of NCL is given by a planar cubic undirected graph $G_{NCL}$ where each edge is colored either red or blue.
Each vertex is either incident to three blue edges or incident to two red and one blue edges.
We respectively call such vertices \emph{OR} and \emph{AND} vertices.
An orientation of $G_{NCL}$ must satisfy the constraint that at every vertex $v \in V(G_{NCL})$, at least one blue edge or at least two red edges are oriented towards $v$.
A \emph{move} is an operation that transforms a satisfying orientation to another by reversing the orientation of a single edge.
The problem gives two satisfying orientations $A$ and $B$ of $G_{NCL}$ and asks for a sequence of moves to transform $A$ into $B$.
As in~\cite{akitaya2019reconfiguration}, we subdivide each edge in $G_{NCL}$ obtaining a bipartite graph $G_{NCL}'$ with one part formed by original vertices in $V(G_{NCL})$ and another part formed by degree-2 vertices.
We require that an orientation must additionally satisfy the constraint that each degree-2 vertex $v$ must have an edge oriented towards $v$.
The question of whether there exists a sequence of moves transforming orientation $A'$ into $B'$ of $G_{NCL}'$ remains PSPACE-complete.
We follow the framework in~\cite{akitaya2019reconfiguration} with a few crucial differences.
The main technical challenge is dealing with the slack constraints while maintaining the desired behavior for the gadgets.
We first describe the reduction to instances of our problem with slack equals zero.
We then generalize the proofs.

\subsection{Zero Slack}
\label{sec:hardness-0}

\ifthenelse{\boolean{lncs}}{\later{\section{Omitted Proofs from Section~\ref{sec:hardness-0}}\label{sec:hard-append}}
}{}

In the following reduction, we are given a bipartite instance of NCL given by $(G_{NCL}', A', B')$, and we produce an instance of BR$(G, k, s)$ of the balanced recombination problem consisting of two $(k,s)$-BCP of a planar graph $G$, $\Pi_A$ and $\Pi_B$, with $k = O( |V(G_{NCL})| )$ {\distr s}, and slack $s=0$. 
\ifthenelse{\boolean{lncs}}{
Here give a brief overview of the reduction. Details can be found in Appendix~\ref{sec:hard-append}.

The AND, OR and degree-2 gadgets are shown in Figures~\ref{fig:gadgets}~(a), (b) and (c) respectively.
The green (black) dots are called \emph{heavy} (\emph{light}) vertices and are considered to be weighed with integer weight more than one (equal to one). We can implement weights by attaching an appropriate number of degree-1 vertices to a heavy vertex so that, in order for a $k$-BCP to be connected, whichever {\distr} contains the heavy vertex must also contain all degree-1 vertices attached to it.
Every edge $e\in E(G_{NCL}')$ is represented by two light vertices of $G$, $e^+$ and $e^-$, that belong to two neighboring gadgets as shown in Figures~\ref{fig:gadgets}~(d).
}{}

\later{
\smallskip\noindent\textbf\paragraph{Construction.}
We first describe a building block used in our gadgets, called \emph{heavy vertices}, represented by green dots in Figure~\ref{fig:gadgets}.
Each heavy vertex $q$ is associated with a positive integer \emph{weight} $w(q)$ and represents a vertex that we also call $q$, slightly abusing notation, attached to $w(q)-1$ degree-1 vertices.
The property that we exploit is that whichever {\distr} that contains vertex $q$ must also contain all $(w(q)-1)$ degree-1 vertices attached to it or else the {\distr} containing one of such vertices would be disconnected.
Then, in practice we can consider the $w(q)$ vertices represented by the heavy vertex $q$ as a single vertex $q$ with weight $w(q)$ towards the size of the {\distr} containing it.
The ordinary vertices (i.e., vertices of weight 1) are called \emph{light}.
We now define the AND and OR gadgets corresponding to the vertices of $G_{NCL}$ and the degree2- gadget corresponding to the degree-2 vertices of $G_{NCL}'$.
Let $\alpha$ be a positive integer to be determined.
The \emph{AND gadget} is shown in Figure~\ref{fig:gadgets}~(a) made of 6 light vertices and 3 heavy vertices.
We set $w(v_a) = w(v_b) = \alpha$ and $w(v_c) = 8 \alpha-3$.
The \emph{OR gadget} is shown in Figure~\ref{fig:gadgets}~(b) made of 6 light vertices and 7 heavy vertices.
We set $w(v_a) = w(v_b) = w(v_a') = w(v_b') = w(v_c') = \alpha$, $w(v_c) = 6 \alpha - 3$, and $w(v') = 9 \alpha$.
The \emph{degree-2 gadget} is shown in Figure~\ref{fig:gadgets}~(c) made of 4 light vertices and 2 heavy vertices.
We set $w(v_a) = w(v_b) = 5 \alpha - 1$.
Each edge $e$ of $G_{NCL}'$ is represented by two vertices, $e^-$ and $e^+$, that are shared by the two gadgets corresponding to the vertices incident to $e$, as shown in Figure~\ref{fig:gadgets}~(d).
For each vertex of $G_{NCL}'$, create a corresponding gadget identifying the vertex pairs 
$e^-$ and $e^+$ as shown in the figure.
This concludes the construction of $G$.
We set $\alpha = 5$ so that the maximum weight of heavy vertices in AND/OR gadgets is greater than $5\alpha-1$.

\,
}

\begin{figure}[h]
	\centering
	\ifthenelse{\boolean{lncs}}
	{\includegraphics[width=\linewidth]{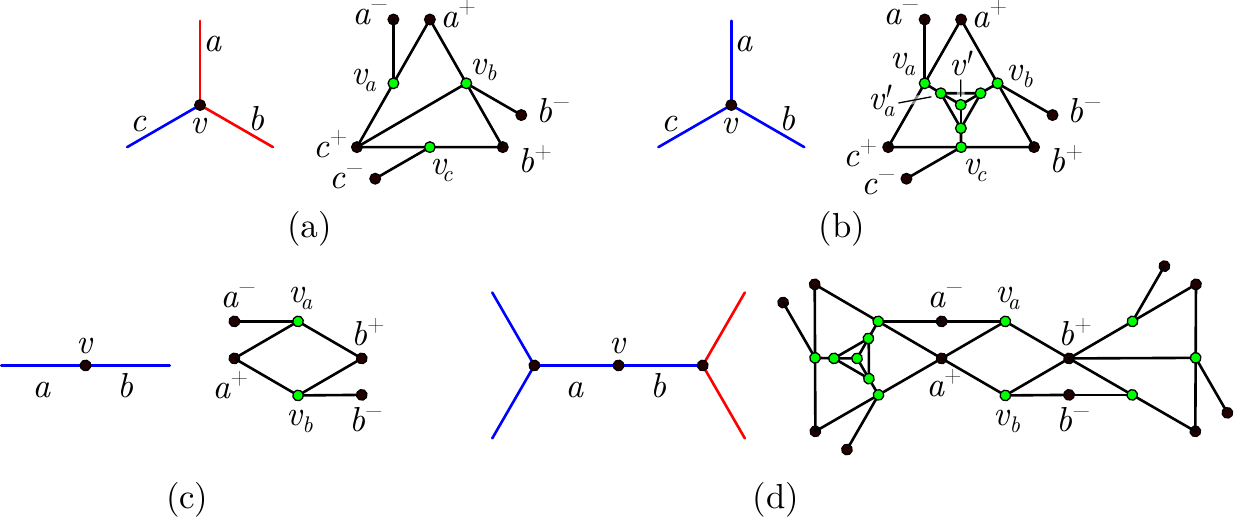}}
	{\includegraphics[width=0.8\linewidth]{gadgets}}
	\caption{Gadgets for the reduction from NCL to BR$(G,k,0)$.}
	\label{fig:gadgets}
\end{figure}

\ifthenelse{\boolean{lncs}}{
The weights of heavy vertices are set up so that, for each AND or OR gadget, a {\distr} must contain its heavy vertices $v_a$, $v_b$ and $v_c$ and no heavy vertices of neighbor degree-2 gadgets.
Then, for all degree-2 gadget, there is a {\distr} that contains $v_a$ and $v_b$ and no other heavy vertex.
Additionally, for all OR gadgets, a {\distr} must contain $v'$ and exactly one vertex in $\{v_a', v_b', v_c'\}$.
We encode whether an edge $e$ points toward a vertex by whether a {\distr} of the corresponding gadget contains $e^+$.
Then, the connectivity constrains of the {\distr s} simulate the NCL constraints. See Figure~\ref{fig:assignment}.
}{}

\later{
We now define $\Pi_A$ and $\Pi_B$ based on $A'$ and $B'$ respectively.
We describe the construction of a set $\mathcal{M}_X$ of $(k,0)$-BCPs of $G$ for an arbitrary satisfying orientation $X$ of $G_{NCL}'$.
Every {\distr} will have $10 \alpha$ vertices.
Refer to Figure~\ref{fig:assignment}.
For every vertex $v \in G_{NCL}'$ create a {\distr} $D_v$ containing $v_e$ for every edge $e$ incident to $v$.
Add $e^+$ ($e^-$) to $D_v$ if $e$ is directed towards (away from) $v$ in $X$.
That concludes the description for AND and degree-2 gadgets.
If $v$ is an OR vertex, let $a$, $b$ and $c$ be the edges incident to $v$.
Since $X$ satisfies the NCL constraints, at least one edge is directed towards $v$.
Without loss of generality, let $a$ be directed towards $v$ in $X$.
Let $v_1$ be an arbitrary vertex in $\{v_a', v_b'\}$.
Then, add the vertices in $\{v_a', v_b', v_c'\} \setminus \{v_1\}$ to $D_v$ and create a {\distr} $D_{v'} = \{v', v_1\}$.
Note that if there are multiple edges directed towards $v$ in $X$, then there are multiple options of edges to be labeled ``$a$", and that $v_1$ is chosen arbitrarily.
Let $\mathcal{M}_X$ be the set of all $k$-partitions that can be constructed from $X$ with the procedure above.
Lemma~\ref{lem:gadget-correspondence} will show that the members of $\mathcal{M}_X$ are indeed $(k,0)$-BCPs.
We choose $\Pi_A \in \mathcal{M}_{A'}$ and $\Pi_B \in \mathcal{M}_{B'}$ arbitrarily.
By construction, $k = |V(G_{NCL}')| + n_{OR}$, where $n_{OR}$ is the number of OR vertices.
This concludes the construction.
Note that given a $k$-partition $\Pi \in \mathcal{M}_{X}$, the above description implies a method to obtain a corresponding orientation $X$ of $ G_{NCL}' $.
Namely, for every degree-2 (AND or OR) gadget, if there is a {\distr} $D_v$ in $\Pi$ containing both $\{v_a, v_b\}$ (all vertices in $\{v_a, v_b, v_c\}$) orient $e$ towards $v$ if and only if $e^+ \in D_v$, $e \in \{a, b\}$ ($e \in \{a, b, c\}$).

\, 
}


\begin{figure}[h]
	\centering
	\ifthenelse{\boolean{lncs}}
	{\includegraphics[width=0.8\linewidth]{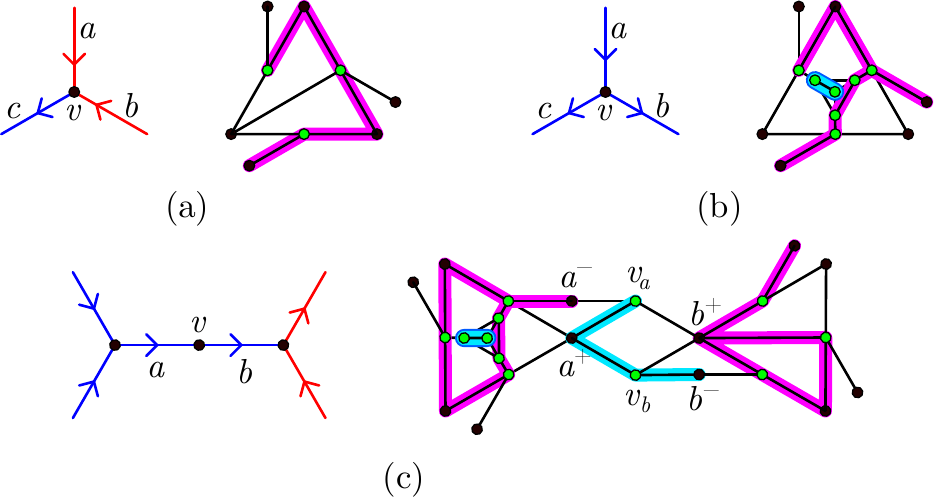}}
	{\includegraphics[width=0.7\linewidth]{assignment}}
	\caption{Equivalence between a satisfying orientation of $G_{NCL}'$ and a $k$-BCP of $G$.}
	\label{fig:assignment}
\end{figure}

\later{
\smallskip\noindent\textbf{Correctness.}
We show that there exist a correspondence between a move in an orientation of $G_{NCL}'$ and a recombination in a $(k,0)$-BCP of $G$.

\ifthenelse{\boolean{lncs}}{}{\newpage}

\begin{lemma}
	\label{lem:gadget-correspondence}
	A $k$-partition $\Pi \in \mathcal{M}_X$ is a $(k,0)$-BCP if and only if $X$ satisfies the NCL constraints.
\end{lemma}

\begin{proof}
	Consider an AND vertex $v$ incident to edges $a$, $b$, and $c$ as in Figure~\ref{fig:gadgets}~(a).
	By definition $D_v$ contains $v_e$ and exactly one of $\{e^+, e^-\}$ for $e \in \{a, b, c\}$.
	If $G[D_v]$ is connected, then $D_v$ must contain either $c^+$, or $a^+$ and $b^-$.
	Then, either $c$ is directed towards $v$ or both $a$ and $b$ are directed towards $v$ in $X$.
	Thus the NCL constraint for $v$ is satisfied.
	The other direction is analogous.
	
	Similarly, we can show the equivalence for degree-2 vertices.
	Consider an OR vertex $v$ incident to edges $a$, $b$, and $c$ as in Figure~\ref{fig:gadgets}~(b).
	Recall that $D_{v'}$ contains $v'$ and exactly one heavy vertex in $\{v_a', v_b', v_c'\}$.
	Then, $G[D_v]$ contains a single path between two vertices in  $\{v_a, v_b, v_c\}$ using vertices in $\{v_a', v_b', v_c'\}$, i.e., $D_v$ induces two components in $\{v_a, v_b, v_c, v_a', v_b', v_c'\}$.
	Since $\{v_a, v_b, v_c\}\subset D_v$ by definition, $D_v$ must contain at least one vertex in $\{a^+, b^+, c^+\}$.
	Thus the NCL constraint for $v$ is satisfied.
	The other direction is analogous.\ifthenelse{\boolean{lncs}}{\hfill$\Box$}{}
\end{proof}

The following lemma shows that the subgraph of the configuration space $\mathcal{R}_0(G,k)$ induced by $ \mathcal{M}_X $ is connected if $X$ satisfies the NCL constraints.

\begin{lemma}
	\label{lem:M-connected}
	Given a satisfying orientation $X$ of $G_{NCL}'$ and two $k$-partitions $\Pi_1, \Pi_2 \in \mathcal{M}_X $, then there is a sequence of at most $n_{OR}$ recombinations that takes $\Pi_1$ to $\Pi_2$.
\end{lemma}
\begin{proof}
	By construction, $\Pi_1$ and $\Pi_2$ are identical except at gadgets of OR vertices.
	Let $v$ be an OR vertex where $\Pi_1(v') \neq \Pi_2(v')$.
	Then recombine $\Pi_1(v')$ and $\Pi_1(v_a)$ into $\Pi_2(v')$ and $\Pi_2(v_a)$.
	The resulting $k$-partition is in $\mathcal{M}_X$ and by Lemma~\ref{lem:gadget-correspondence} its {\distr s} induce connected subgraphs.
	By repeating this procedure at all OR gadgets whose {\distr s} differ from the ones in $\Pi_2$ we eventually get $\Pi_2$.
	There can be at most $n_{OR}$ recombinations since this is the number of OR gadgets.\ifthenelse{\boolean{lncs}}{\hfill$\Box$}{}
\end{proof}

The following lemma establishes the equivalence between a recombination move in $G$ and an edge reversal move in $G_{NCL}'$.

\begin{lemma}
	\label{lem:flip-recomb}
	For every $(k,0)$-BCP $\Pi_1 \in \mathcal{M}_{X_1}$ obtained from a satisfying orientation $X_1$ on $G_{NCL}'$, every recombination on $\Pi_1$ yields a $(k,0)$-BCP $\Pi_2$ such that $\Pi_2 \in \mathcal{M}_{X_2}$ where $X_2$ is an orientation on $G_{NCL}'$ that differs from $X_1$ by the orientation of at most a single edge.
	Similarly, for an orientation $X_2$ on $G_{NCL}'$ obtained from $X_1$ by reversing the orientation of a single edge, there is a sequence of one or two recombinations that takes any $\Pi_1 \in \mathcal{M}_{X_1}$ to some $\Pi_2 \in \mathcal{M}_{X_2}$.
\end{lemma}
\begin{proof}
	We begin with proving the first claim.
	Let $D_p$ and $D_q$ be two {\distr s} in $\Pi_1$ 
	that are recombined into {\distr s} $V$ and $W$ of $\Pi_2$.
	First assume that $p=v$ is an OR vertex of $G_{NCL}'$, $D_v$ is the {\distr} containing $\{v_a, v_b, v_c\}$, and $q=v'$.
	Let $W$ be the {\distr} containing $v'$.
	Since $w(v') = 9 \alpha$ and $n/k = 10 \alpha$, then $W$ must contain exactly $v'$ and one heavy vertex in $\{v_a', v_b', v_c'\}$.
	Then $D_p$ and $V$ induces the same graph on $\{v_a, v_b, v_c, a^+, b^+, c^+\}$.
	Thus, we can read a graph orientation $X_2 = X_1$ from $\Pi_{2}$, i.e., both $\Pi_{1}$ and $\Pi_{2}$ are in $ \mathcal{M}_{X_1} $.
	Now, assume that $p$ is an AND/OR vertex of $G_{NCL}'$ adjacent to edges $a$, $b$, and $c$, $D_p$ is the {\distr} containing $\{p_a, p_b, p_c\}$, and $q$ is a degree-2 vertex of $G_{NCL}'$ adjacent to edges $c$ and $d$ with $D_q$ being the {\distr} containing $\{p_c, p_d\}$.
	Let $V$ be the {\distr} containing the heavier  vertex in $\{p_a, p_b, p_c\}$.
	By construction, the weight of such vertex is greater than $5 \alpha - 1$.
	Then $V$ cannot contain $\{p_c, p_d\}$ and $\{p_c, p_d\} \subset W$.
	Due to size constraints, $W$ cannot contain any other heavy vertex and $\{p_a, p_b, p_c\} \subset V$.
	Note that $D_p$ and $V$ ($D_q$ and $W$) differ only by a vertex in $\{c^-, c^+\}$.
	We can then read an orientation $X_2$ off of $\Pi_{2}$ that differs from $X_1$ only by the orientation of $c$.
	
	We proceed to prove the second claim.
	Let $c = pq \in E(G_{NCL}')$ be the edge whose orientation differs between $X_1$ and $X_2$.
	Let $p$ be an AND/OR vertex incident to $a, b, c \in E(G_{NCL}')$, and $q$ be a degree-2 vertex incident to $c, d \in E(G_{NCL}')$. 	
	Let $\Pi_1$ be an arbitrary $(k,0)$-BCP in $\mathcal{M}_{X_1}$.
	First, assume that $p$ is an AND vertex.
	Recombine $D_p, D_q \in \Pi_{X_1}$ into $D_p' = D_p \setminus \{c^-\} \cup \{c^+\}$ and $D_q' = D_q \setminus \{c^+\} \cup \{c^-\}$ obtaining $\Pi_{2}$.
	By Lemma~\ref{lem:gadget-correspondence}, the {\distr s} in $\Pi_{2}$ are connected and by construction $\Pi_{2} \in \mathcal{M}_{X_2}$.
	Now assume that $p$ is an OR vertex.
	If $c$ is directed towards $q$ in $X_1$, the same proof as above holds.
	Else, $c$ is directed towards $p$ in $X_1$ and since $X_2$ satisfies the NCL constraints there is an edge $e \in \{a, b\}$ directed towards $p$.
	If $D_{p'} \in \Pi_{X_1}$ does not contain $v_e'$, we can perform one recombination so that it does as in Lemma~\ref{lem:M-connected}.
	The remainder of the proof is the same as for the case when $p$ is an AND vertex.\ifthenelse{\boolean{lncs}}{\hfill$\Box$}{}
\end{proof}

The combination of Lemmas~\ref{lem:gadget-correspondence}, \ref{lem:M-connected}, and \ref{lem:flip-recomb} yield the following lemma.

}

\both{
\begin{lemma}
	\label{lem:hardness}
	BR$(G,k,0)$ is PSPACE-complete even for a planar graph $G$ with constant maximum degree.
\end{lemma}
}

\subsection{Generalizations}
\label{sec:hardness-gen}

\ifthenelse{\boolean{lncs}}{
We generalize the reduction of Lemma~\ref{lem:hardness}.
Details are in Appendix~\ref{sec:app-hardness-gen}.

\smallskip\noindent\textbf{Bounded-degree triangulation $G$.}
The main new technical tool presented in this section is the \emph{filler gadget} shown in Figure~\ref{fig:3-connected}~(b).
Each face marked with a dot is called a \emph{heavy face} associated with an integer weight, and whose recursive construction is shown in Figure~\ref{fig:3-connected}~(a). 
Figure~\ref{fig:3-connected}~(c) shows how to use copies of the filler gadget to transform $G$ in a triangulation. 
The main property of the filler gadget is that we set the weights of heavy faces so that each red vertex must belong to a different {\distr} and the gadget only intersects 5 {\distr s}.
Then, such districts are ``trapped" in the filler gadget and don't interfere with the other gadgets.

\later{\section{Omitted Proofs from Section~\ref{sec:hardness-gen}}
\label{sec:app-hardness-gen}}
}{}

\later{
This section presents three generalizations of Lemma~\ref{lem:hardness}.
The second generalization subsumes the first as well as Lemma~\ref{lem:hardness}. 
However we present all three in this order for ease of exposition.

\smallskip\noindent\textbf{Larger slack.}
By simply adjusting the value of $\alpha = 5 + s$, we can generalize Lemma~\ref{lem:hardness} to larger slacks $s = O(n^{1-\eps})$ for any constant $0 < \eps \le 1$.
Indeed, any value $\alpha > 4 + s$ would work.
There is a trade-off between the size of the slack and the maximum degree of $G$ as the maximum degree is $\Theta(\alpha)$. Note the orientation of an edge $e$ is encoded by the membership of the corresponding vertex $e^+$ in the $(k,s)$-BCP of $G$.
In Section~\ref{sec:hardness-0}, $e^-$ was used to balance the size of the {\distr s};
and $e^-$ and $e^+$ were necessarily in a different {\distr s}.
When the slack $s$ is positive, this will no longer be the case, but as long as the {\distr s} in each gadget contain the same heavy vertices as in Section~\ref{sec:hardness-0}, we can use membership of $e^+$ to encode the orientation of $e$.
The key properties that we have to maintain for all $(k,s)$-BCPs $\Pi$ are the following:
\begin{enumerate}\itemsep 0pt
\item[(i)] for all $v \in V(G_{NCL}')$ of degree-2 incident to edges $a$ and $b$, $\Pi(v_a)=\Pi(v_b)$;
\item[(ii)] for $v \in V(G_{NCL}')$ of all degree-3 incident to edges $a$, $b$, and $c$, $\Pi(v_a)=\Pi(v_b)=\Pi(v_c)$; and
\item[(iii)] for all OR vertices $v \in V(G_{NCL}')$, $\Pi(v')$ contains exactly one heavy vertex other than $v'$.
\end{enumerate}

Property (iii) is clearly satisfied because $n/k$ remains $10 \alpha$ and $\Pi(v')$ cannot contain only $v'$ or else it would contain less than $n/k - s = 9 \alpha + 5$ vertices, and it also cannot contain 3 heavy vertices or it would contain more than $n/k + s = 11 \alpha - 5$ vertices.
Consider Property (ii).
Recall that $w(v_c) \ge 6 \alpha - 3$.
Then, $\Pi(v_c)$ cannot contain any heavy vertices of adjacent gadgets (of weight $5 \alpha - 1$) or it would be larger than the maximum allowed size.
Then $\Pi(v_c)$ must contain only heavy vertices in the gadget.
Indeed, all other available heavy vertices have weight $\alpha$.
Then, $\Pi(v_c)$ must contain all other available heavy vertices in the gadget (excluding $v'$ and one other in case $v$ is an OR vertex because of (iii)) or it would contain less then the minimum allowed  number of vertices.
Then, Property (i) becomes trivial given (ii) and (iii).
Note that this approach does not allow the slack to be linear in $n$.
By construction, the size of the instance becomes $n = |V(G)| = \Theta(s|V(G_{NCL})|)$.
We allow $s$ to be any polynomial in $|V(G_{NCL})|$ as long as the size of the output instance of the reduction is polynomial.
Note that $k$ is fixed for a given $G_{NCL}$. Then $k$ is bounded by $O(n/s)$.

\begin{corollary}
	\label{cor:hardness-slack}
	BR$(G,k,s)$ is PSPACE-complete even if $G$ is planar, $k \in O(n^\eps)$, and $s \in O(n^{1-\eps})$ for any constant $0 < \eps \le 1$.
\end{corollary}

\smallskip\noindent\textbf{Bounded-degree triangulation $G$.}
We now generalize Corollary~\ref{cor:hardness-slack} restricting $G$ to be a planar triangulation of constant maximum degree. 
As a consequence, our reduction proves that BCP is also NP-hard for this class of graphs generalizing the results in~\cite{WakabayashiCS07}.
The main technical ingredient of this section is the gadget shown in Figure~\ref{fig:3-connected}~(b) called \emph{filler gadget}.
It consists of 5 vertices labeled ``red'' inducing a maximal planar subgraph.
Each internal face $f$ of this induced subgraph is called a \emph{heavy face}, and is assigned a positive integer weight $w(f)$.
We assume that $w(f)$ is a multiple of 3.
Each heavy face represents a subgraph $G_f$ with $w(f)+3$ vertices whose outer face is $f$.
Figure~\ref{fig:3-connected}~(a) shows the recursive construction of $G_f$.
Then $G_f$ is the 4-connected maximal planar graph with 6 vertices (including the 3 vertices of $f$) containing a heavy face $f'$.
The face $f'$ is chosen so that it is vertex-disjoint from $f$ and $w(f') = w(f) - 3$.
The base case, a heavy face with zero weight, is simply an ordinary face.

}

\begin{figure}[h]
	\centering
	\includegraphics[width=\linewidth]{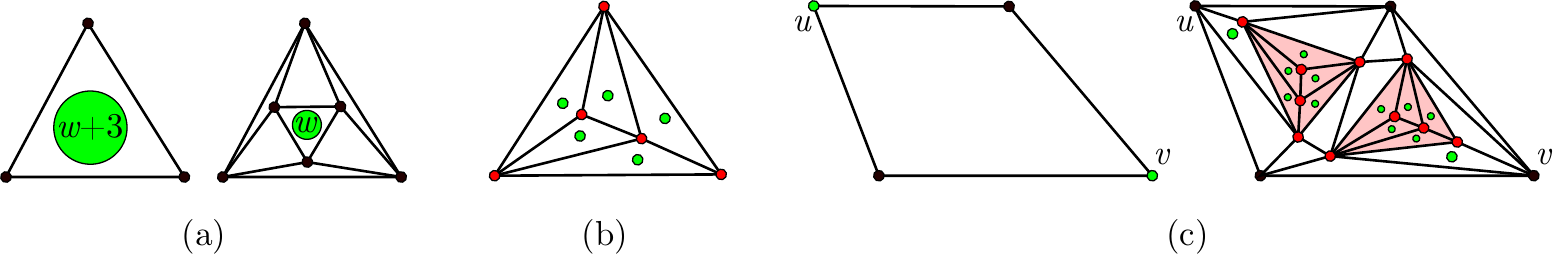}
	\caption{Construction of the filler gadget.}
	\label{fig:3-connected}
\end{figure}

\later{
We now describe how to modify $G$ produced in the reduction in Section~\ref{sec:hardness-0}.
We will add filler gadgets to the faces of $G$ and triangulate the resulting plane graph as follows.
We set the weight of every heavy face in the filler gadgets to $10 \alpha - s - 1$.
We proceed with the details; refer to Figure~\ref{fig:3-connected}~(c).
First, assign each heavy vertex to a face of $G$.
For each face $f$ assigned with one or more heavy vertices, add one copy of the filler gadget to the interior of $f$ for each heavy vertex.
Triangulate $f$ so that there is no new edge between vertices of $f$, i.e., all new edges have at least one endpoint at a filler gadget.
We also require that for each heavy vertex $v$ there is a face $f_v$ containing $v$ and one edge of the outer face of its corresponding filler gadget $F_v$.
We transfer the weight from $v$ to $f_v$ by making $v$ a light vertex and $f_v$ a heavy face, and setting $w(f_v) = 5s + w(v) - 1$.
For remaining faces $f$ of $G$, if $f$ is not a triangle and has not been assigned any heavy vertex, add a copy of the filler gadget in $f$ triangulating it as before, i.e., no new edge should be between vertices of $f$.
Choose a neighboring face $f'$ (sharing an edge with the gadget) to make if heavy with $w(f') = 5s$.
The resulting graph $G'$ is an edge-maximal planar graph.
Finally, set $\alpha = 11s + 5$ and increase $k$ by 5 for each copy of the filler gadget used.

It remains to define the initial and target $(k,s)$-BCPs $\Pi_A$ and $\Pi_B$ defining an instance of BR$(G,k,s)$.
We construct such $(k,s)$-BCPs with $0$ slack.
For that, we show that we can find a balanced connected partition of the filler gadget and $5s$ vertices in the neighbor heavy face into 5 connected {\distr s}.
We construct a Hamiltonian path of the gadget and its neighboring heavy face as follows.
The main structure of the Hamiltonian path is shown in Figure~\ref{fig:Ham-path}~(a).
For each heavy face we need to find a path that visits all internal vertices and starts and ends with a vertex on its outer face.
We accomplish this with the recursive construction shown in Figure~\ref{fig:Ham-path}~(b).
Considering that the path starts in the filler gadget, we ignore the first edge to obtain the desired Hamiltonian path.
It is now easy to partition the $10\alpha$-long prefix of the path into 5 connected components.
The remainder of the construction follows Section~\ref{sec:hardness-0}.

\begin{figure}[h]
	\centering
	\includegraphics[width=0.6\linewidth]{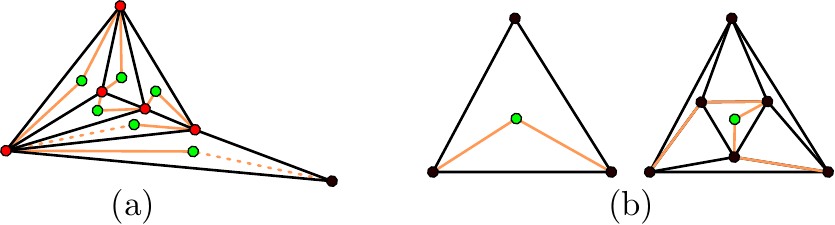}
	\caption{Inductive construction of the Hamiltonian path used to partition the filler gadget.}
	\label{fig:Ham-path}
\end{figure}

\begin{lemma}
	\label{lem:filler-gadget}
	In every $(k,s)$-BCP $\Pi$ of $G'$, for every filler gadget $F$, exactly five {\distr s} of $\Pi$ intersect $F$ and the number of vertices in the union of these five {\distr s} outside of $F$ is at most $10s$.
\end{lemma}
\begin{proof}
	First note that $n/k$ remains $10\alpha$.
	The number of vertices in $F$ is $5(10\alpha - s)$, thus it must intersect at least 5 {\distr s} of $\Pi$.
	Each {\distr} that intersects $F$ must contain a red vertex since each heavy face has less than the minimum capacity of a {\distr}.
	Since $F$ has only 5 red vertices, there must be exactly 5 {\distr s} intersecting $F$, each containing exactly one red vertex.
	The maximum number of vertices of these {\distr s} outside of $F$ is $5(10\alpha + s) - 5(10\alpha - s) = 10s$.\ifthenelse{\boolean{lncs}}{\hfill$\Box$}{}
\end{proof}

Lemma~\ref{lem:filler-gadget} allows us to mostly ignore filler gadgets.
We refer to the {\distr s} that intersect a filler gadget as a \emph{filler {\distr}}.
Recall that every heavy vertex $v$ in $G$ is now adjacent to a heavy face $f_v$ in $G'$.
Since $10s < w(f_v) < 10\alpha -s$, if a filler {\distr} contains $v$ then it would disconnect vertices of $f_v$ that cannot be in any {\distr}.
Then, the {\distr} that contains $v$ is not filler and it must also contain the $w(v) \pm 5s - 1$ vertices in $f_v$ that are not in filler {\distr s}.
That mimics the weights of heavy vertices in $G$ such as $v$ in a 3-connected graph $G'$ within a $\pm 5s$ accuracy.
One can check that the new choice of $\alpha$ accommodates this variations guaranteeing properties (i)--(iii).
We omit the details since the argument is fairly similar as before.
The situation in which filler {\distr s} contain normal vertices of $G$ would then be equivalent to the reduction in Section~\ref{sec:hardness-0} where some of the non-heavy vertices are deleted.
Intuitively, deleting vertices only makes it more difficult to satisfy the connectivity constraints.
We need to extend the equivalence to the NCL orientation to allow for when vertices $e^+$ corresponding to an edge $e$ are taken by filler {\distr s}.
That would correspond to a partial orientation where $e$ is not directed, i.e., it does not count for satisfying the constraints of neither its endpoints.
It is clear that such a variant of NCL remains PSPACE-complete since we still require that the constraints are satisfied at each vertex ignoring undirected edges and a sequence of operations in this setting can be converted by a sequence in the normal NCL setting and vice versa by assigning an arbitrary direction for an undirected edge.
For more details see an asynchronous vertion of NCL~\cite{Vig13}.
Then, Lemmas~\ref{lem:gadget-correspondence}--\ref{lem:flip-recomb} hold for this variant.
The construction distributes the weights of heavy vertices of $G$ into heavy faces of $G'$.
By the recursive construction, the degree of each vertex increases by at most 2.
Then, the filler gadget has constant degree.
Note that $k$ remains $O(n/s)$.

\,
}
\both{
\begin{theorem}
	\label{thm-hard-3-con}
	BR$(G,k,s)$ is PSPACE-complete even if $G$ is maximal planar of constant maximum degree, $k \in O(n^\eps)$, and $s \in O(n^{1-\eps})$ for $0 < \eps \le 1$.
\end{theorem}
}

\smallskip\noindent\textbf{Finding Balanced Connected Partitions.}
The \emph{NCL orientation problem} is defined by an input undirected graph $G_{NCL}$ edge colored as before, and asks whether there exist an orientation of $G_{NCL}$ that satisfies the NCL constraints.
This problem in NP-complete~\cite{PSPACE-book}.
We remark that our construction implies the following theorem.

\begin{theorem}
	\label{bcp-hard-3-con}
	It is NP-complete to decide whether there exist a $(k,s)$-BCP of a graph $G$, even if $G$ is maximal planar of constant maximum degree, $k \in O(n^\eps)$, and $s \in O(n^{1-\eps})$ for $0 < \eps \le 1$.
\end{theorem}

\ifthenelse{\boolean{lncs}}{
\smallskip\noindent\textbf{Constant number of {\distr s}.}
The drawback of the previous construction is that it requires 5 new districts for each filler gadget. Here we obtain PSPACE-hardness with $k=3$, but we lose the restriction that $G$ is a triangulation, and instead we only require that $G$ is a bounded-degree planar graph.
The main technical difficulty is to guarantee that the same subset of heavy vertices is always contained in the same {\distr}. This property is obtained by a careful setting of the weights so that there is a unique partition of weights of the heavy vertices that allow for the 3 {\distr s} to be balanced within $s$ slack.
This allow us to label the districts according to the heavy vertices that is contains.
One of the districts then locally acts like the {\distr s} that previously contained $v_a$, $v_b$ and $v_3$ in each AND and OR gadgets, maintaining the equivalency between the conectedness of this district and the NCL constraints.
}{}

\later{
\smallskip\noindent\textbf{Constant number of {\distr s}.}
Note that, in the previous hardness proofs, $k$ cannot be a constant.
We now build on Lemma~\ref{lem:hardness} to prove PSPACE-hardness for instances of BR$(G,3,s)$.
Note that with $k = 2$ the configuration space is the complete graph by definition.
Also note that this result does not subsume Theorem~\ref{thm-hard-3-con} since the graph produced has cut-vertices as we will see in the remainder of this section.
We now show how to modify the graph $G$ produced in the reduction of Section~\ref{sec:hardness-0}.
We modify the gadgets as follows.
The AND gadget remains unchanged while the OR gadget adjacent to edges $a$, $b$ and $c$ as in Figure~\ref{fig:constant-districts-planar}~(a) is changed by  adding the edges $a^+v_a'$, $a^+v_b'$, $b^+v_b'$, $b^+v_c'$, $c^+v_a'$, and $c^+v_c'$.
The vertices $v_a'$, $v_b'$ and $v_c'$ are no longer heavy vertices.
The main difference is in the degree-2 gadgets which are shown in Figure~\ref{fig:constant-districts-planar}~(b).
They now have six heavy vertices each, $v_a$, $v_a'$, $v_{ab}$, $v_{ab}'$, $v_b$, $v_b'$.
Refer to Figure~\ref{fig:constant-districts-planar}~(c)--(d).
For each face $f$ of $G_{NCL}$ we create a new heavy vertex $v_f$ in $G$ which concludes the description of $V(G)$.
We now describe how to complete $E(G)$.
For that we have to describe two spanning trees of a new graph $G_{NCL}^*$ described as follows.
We build $G_{NCL}^*$ from $G_{NCL}$ by adding $v_f$ to the face $f$, triangulating the result by adding edges incident to $v_f$, and deleting every original edge of $G_{NCL}$.
Refer to Figure~\ref{fig:constant-districts-planar}~(c).
Note that $G_{NCL}^*$ is bipartite, there is a one-to-one correspondence between degree-2 vertices of $G_{NCL}'$ and faces of $G_{NCL}^*$, and that such faces are quadrilaterals.
Let $T_1$ be a spanning tree of $G_{NCL}^*$ and $T_2$ be an interdigitating spanning tree of the dual of $G_{NCL}^*$, i.e., spanning all degree-2 vertices of $G_{NCL}'$ so that $T_1$ and $T_2$ do not both use the the primal or dual of the same edge.
For every edge $u v_f$ in $T_1$, add an edge $u_a v_f$, where $a$ is an edge incident to $u$ in $G_{NCL}'$ and $a$ appears before $u$ in a clockwise traversal of $f$.
For every edge $u v$ in $T_2$, add the edge $u_{ab}'v_{cd}$ or $u_{ab} v_{cd}'$ where $a$ and $b$ (resp., $c$ and $d$) are the edges incident to $u$ (resp., $v$) and the addition of the edge does not introduce a crossing in $G$.
Then, for every pair of vertices of the form $v_a$ or $v_a'$ where $a$ is an edge incident to a degree-2 vertex $v$ in $G_{NCL}'$ that are currently in the same face of $G$, add an edge between them excluding edges of the form $v_a v_a'$, i.e., an edge inside a face of a degree-2 vertex.
Note that the new edges could be of the form $v_a v_b$, i.e., connecting two vertices in the same degree-2 vertex but corresponding to a different edge of $G_{NCL}'$ as shown in Figure~\ref{fig:constant-districts-planar}~(d) by curved edges.
Finally, delete one of the edges added in the last step.

We now give some informal intuition about the construction.
We will set up the weights so that heavy vertices can only be part of the same {\distr s} which we label $V_1$, $V_2$ and $V_3$ in any $(k,s)$-BCP of $G$.
The following is the key property that we exploit.

\begin{enumerate}[label=$(\star)$]
	\item \label{prop:heavyfixed}
	{\Distr} $V_1$ must contain every vertex $v_f$ created from a face $f$ of $G_{NCL}$ as well as every $v_a$, $v_b$ and $v_c$ of a gadget corresponding to a degree-3 vertex $v$ of $G_{NCL}'$.
	{\Distr} $V_2$ must contain all vertices of degree-2 gadgets of the form $v_{ab}$ or $v_{ab}'$.
	{\Distr} $V_3$ must contain all remaining heavy vertices, i.e., heavy vertices of degree-2 gadgets of the form $v_{a}$ and heavy vertices $v'$ of gadgets corresponding to OR vertices $v$.
\end{enumerate}

The connectivity of $V_1$ ($V_2$) will mimic, in broad terms, the spanning tree $T_1$ ($T_2$).
{\Distr} $V_3$ will separate $V_1$ and $V_2$ also containing the heavy vertex $v'$ of OR gadgets.
We now encode the direction of an edge $e$ in $G_{NCL}'$ by whether $V_1$ or $V_3$ contain $e^+$ respectively meaning that $e$ points to its degree-3 or degree-2 end.

We now adjust the weights of the heavy vertices of $G$ in order to obtain property~\ref{prop:heavyfixed}.
Let $n_2$, $n_{OR}$ and $n_{AND}$ be the number of degree-2, OR and AND vertices in $G_{NCL}'$ respectively.
Set the weights of $V_1$'s heavy vertices to $w_1 = \alpha$, which we will set later.
Choose one of the heavy vertices assigned to $V_2$ and set its weight to $n/3 + s - w_1$, and set the weights of its remaining heavy vertices to $w_2 = (w_1 - s - n_2) / (2 n_2 - 1)$.
Choose one of the heavy vertices assigned to $V_3$ and set its weight to $n/3 + s - w2$, and set the weights of its remaining heavy vertices to $w_3 = (w_2 - s - 3 n_2 - n_{OR}) / (4 n_2 + n_{OR} - 1)$.
We set $\alpha$ so that $w_3 = s+n_\ell$ where $n_\ell$ is the number of light vertices in $G$.
Then $\alpha \in O(n_2^2 s)$ and $|V(G)| \in O(\alpha |V(G_{NCL})|)$.

\begin{figure}[h]
	\centering
	\ifthenelse{\boolean{lncs}}
	{\includegraphics[width=\linewidth]{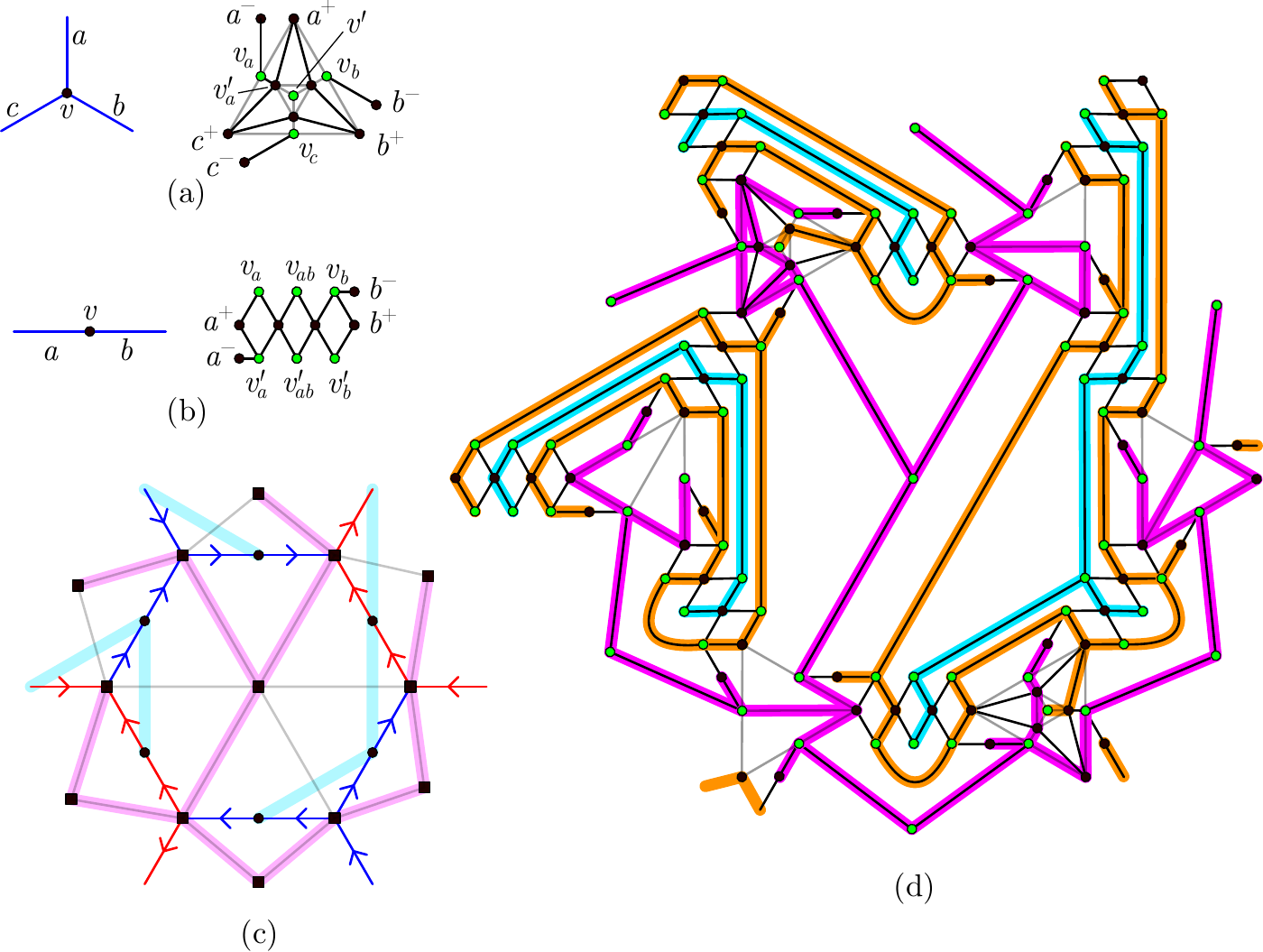}}
	{\includegraphics[width=0.9\linewidth]{constant-districts-planar.pdf}}
	\caption{Construction of the reduction with $k=3$. The three {\distr s} $v_1$, $v_2$, and $V_3$ are colored magenta, cyan, and orange respectively. (a) New OR gadget. (b) New degree-2 gadget. (c) $G_{NCL}^*$ overlapped with $G_{NCL}'$. The edges of $G_{NCL}^*$ are shown in gray and its vertices are shown as squares. (d) The corespondent part of $G$.}
	\label{fig:constant-districts-planar}
\end{figure}
}
\both{
\begin{theorem}
	\label{thm:hardness-constant-districts-planar}
	BR$(G,3,s)$ is PSPACE-complete even if $G$ is  planar with constant maximum degree, and $s \in O(n^{1-\eps})$ for $0 < \eps \le 1$.
\end{theorem}
}

\later{
\begin{proof}
	As before, each individual weight of a heavy vertex is smaller than the minimum threshold for a {\distr} and thus a {\distr} must contain all vertices in the heavy vertex.
	We can implement the weights of a heavy vertex $v$ with a path of length $w(v)$ with an endpoint at $v$ in order to keep the maximum degree of $G$ upper bounded by a constant.
	We can assume that $G_{NCL}$ has constant face degree, or else one can use techniques by Hearn and Demaine~\cite{PSPACE-book} to split large faces using their ``red-blue conversion" gadget.
	Thus, we keep the maximum degree of $G$ upper bounded by a constant.
	We first show that property \ref{prop:heavyfixed} holds in any $(k,s)$-BCP $\Pi$ on $G$.
	By construction, we have that $w_1 > w_2 > w_3$.
	We first show the claim for $V_3$.
	Assume that $V_3$ is the {\distr} containing the heavy vertex of wight $n/3 + s - w2$.
	For contradiction assume that it also contain a heavy vertex whose weight is not $w_3$.
	Then such vertex must have a weight larger than $w_2$ which makes the size of $V_3$ greater than $n/3 + s$, a contradiction.
	Now assume for contradiction that $V_3$ does not contain all of the heavy vertices with weight $w_3$.
	Then the maximum size of $V_3$ is $n/3 - w_3 + n_\ell = n/3 - s$ even if $V_3$ contains all light vertices, a contradiction since $V_1$ and $V_2$ need to contain light vertices to be connected.
	We conclude that $V_3$ must satisfy \ref{prop:heavyfixed}.
	The arguments for $V_2$ are similar once \ref{prop:heavyfixed} is established for $V_3$, and thus $V_1$ also satisfies \ref{prop:heavyfixed}.
	
	We now show that Lemma~\ref{lem:gadget-correspondence} holds in this context.
	By \ref{prop:heavyfixed}, $V_1$ must connect all $v_a$, $ v_b $ and $ v_c $ in every AND and OR gadget. Moreover, apart from such vertices, $V_1$ can only contain light vertices and heavy vertices of the form $v_f$ that came from a face $f$ of $G_{NCL}$.
	Thus the only paths that connect $v_a$, $ v_b $ and $ v_c $ in a given gadget are contained in the gadget itself.
	Therefore, the NCL constraints at degree-3 gadgets must be satisfied as before.
	Note that the new edges in the OR gadget do not change its behavior.
	It remains to show the same for degree-2 gadgets.
	By \ref{prop:heavyfixed}, $V_2$ must connect all $v_{ab}$ and $ v_{ab}' $ in every degree-2 gadget.
	By construction and Property~\ref{prop:heavyfixed}, all paths that connect such vertices must be contained in the same gadget.
	Similarly, $V_3$ must connect $v_{a}$ and $v_a'$ ($v_{b}$ and $v_b'$) by a path contained in the degree-2 gadget.
	Then, at least one vertex in $\{a^+, b^+\}$ must be in $V_3$ to allow space for $V_2$ to connect $v_{ab}$ and $ v_{ab}'$.
	
	To conclude the proof it is enough to show that Lemma~\ref{lem:flip-recomb} holds in this context.
	By the way we encode the orientation of edges, only a recombination between $V_1$ and $V_3$ can change the orientation of an edge.
	The potential worry is that, now that $V_1$ and $V_3$ are adjacent at every gadget, one recombination could alter the orientation so that we can't transform it into a sequence of single orientation reversals in $G_{NCL}$.
	However, the presence of $V_2$ in every degree-2 vertex prevents that the orientation of two edges incident to the same degree-2 vertex flip simultaneously.
	Then the set of edges that flip in one recombination induces paths of at most two edges having a degree-3 vertex as a middle vertex.
	Then we can represent the simultaneous flips with a sequence of flips having first all flips that point an edge towards a degree-3 vertex followed by all flips that point an edge towards a degree-2 vertex.
	Each flip in the two subsequences is independent of other flips in the same subsequence.
	As proved by Viglietta~\cite{Vig13}, even though one operation can flip many independent edges simultaneously, the problem is still equivalent to regular NCL.
	In the other direction, we can show that an edge flip can always be obtained by at most 2 recombinations.
	The first recombination is between $V_2$ and $V_3$ so that $V_2$ contains only the appropriate path between each $v_{ab}$ and $ v_{ab}'$. The second recombination is between $V_1$ and $V_3$, exchanging the membership of $a^+$ and $a^-$ representing the reversal of an edge $a$.\ifthenelse{\boolean{lncs}}{\hfill$\Box$}{}
\end{proof}
}

\section{Conclusion and Open Problems}
\label{sec:conclusion}

We have shown that the configuration space $\mathcal{R}_s(G,k)$ of $(k,s)$-BCPs is connected when $G$ is connected and $s=\infty$, or when $G$ is Hamiltonian and $s\geq |V(G)|/k$. 
We hope that our results inform future research on the properties of $G$, $k$, and $s$ that are sufficient to obtain an efficient sampling of $\Bal_s(G,k)$.
We also leave it as an open problem whether our results in Section~\ref{sec:hamiltonian}  generalize to other classes of graphs. We conjecture that the configuration space $\mathcal{R}_s(n,k)$ is connected for every \emph{biconnected} graph $G$ on $n$ vertices when $s\geq n/k$. However, our techniques do not directly generalize; it is unclear how to extend the notion of canonical $k$-partitions in the absence of a Hamilton cycle.

We have shown that BR$(G,k,s)$ is PSPACE-complete even in specific settings that are of interest in applications such as sampling electoral maps.
Our results imply that the configuration space $\mathcal{R}_s(G,k)$ has diameter exponential in $n$, 
establishing as well an exponential lower bound on the mixing time of a Markov chain on $\mathcal{R}_s(G,k)$ for these settings.
We note that Theorems~\ref{thm-hard-3-con} and \ref{thm:hardness-constant-districts-planar} do not include other settings of interest such as when $G$ is maximal planar (or even 3-connected) and $k$ is a constant.
We leave these as open problems.

\bibliographystyle{plainurl}
\bibliography{redistricting}

\begin{thebibliography}{10}

\bibitem{AGR+19}
Tara Abrishami, Nestor Guillen, Parker Rule, Zachary Schutzman, Justin Solomon,
  Thomas Weighill, and Si~Wu.
\newblock Geometry of graph partitions via optimal transport.
\newblock {\em Preprint}, 2019.
\newblock \href {http://arxiv.org/abs/1910.09618} {\path{arXiv:1910.09618}}.

\bibitem{akitaya2019reconfiguration}
Hugo~A. Akitaya, Matthew~D. Jones, Matias Korman, Christopher Meierfrankenfeld,
  Michael~J. Munje, Diane~L. Souvaine, Michael Thramann, and Csaba~D.
  T{\'{o}}th.
\newblock Reconfiguration of connected graph partitions.
\newblock {\em Preprint}, 2019.
\newblock \href {http://arxiv.org/abs/1902.10765} {\path{arXiv:1902.10765}}.

\bibitem{WakabayashiCS07}
Fr{\'{e}}d{\'{e}}ric Chataigner, Liliane~Benning Salgado, and Yoshiko
  Wakabayashi.
\newblock Approximation and inapproximability results on balanced connected
  partitions of graphs.
\newblock {\em Discret. Math. Theor. Comput. Sci.}, 9(1), 2007.
\newblock URL: \url{http://dmtcs.episciences.org/384}.

\bibitem{Chlebikova96}
Janka Chleb{\'{\i}}kov{\'{a}}.
\newblock Approximating the maximally balanced connected partition problem in
  graphs.
\newblock {\em Inf. Process. Lett.}, 60(5):223--230, 1996.
\newblock \href {https://doi.org/10.1016/S0020-0190(96)00175-5}
  {\path{doi:10.1016/S0020-0190(96)00175-5}}.

\bibitem{DDS19}
Daryl~R. DeFord, Moon Duchin, and Justin Solomon.
\newblock Recombination: {A} family of {M}arkov chains for redistricting.
\newblock {\em Preprint}, 2019.
\newblock \href {http://arxiv.org/abs/1911.05725} {\path{arXiv:1911.05725}}.

\bibitem{duchin2018gerrymandering}
Moon Duchin.
\newblock Gerrymandering metrics: {H}ow to measure? {W}hat's the baseline?
\newblock {\em Preprint}, 2018.
\newblock \href {http://arxiv.org/abs/1801.02064} {\path{arXiv:1801.02064}}.

\bibitem{DyerF85}
Martin~E. Dyer and Alan~M. Frieze.
\newblock On the complexity of partitioning graphs into connected subgraphs.
\newblock {\em Discret. Appl. Math.}, 10(2):139--153, 1985.
\newblock \href {https://doi.org/10.1016/0166-218X(85)90008-3}
  {\path{doi:10.1016/0166-218X(85)90008-3}}.

\bibitem{PrincetonStudy}
Benjamin Fifield, Michael Higgins, Kosuke Imai, and Alexander Tarr.
\newblock Automated redistricting simulation using {M}arkov chain {M}onte
  {C}arlo.
\newblock {\em Journal of Computational and Graphical Statistics}, pages 1--14,
  2020.
\newblock \href {https://doi.org/10.1080/10618600.2020.1739532}
  {\path{doi:10.1080/10618600.2020.1739532}}.

\bibitem{Gyori76}
Ervin Gy\H{o}ri.
\newblock On division of graphs to connected subgraphs.
\newblock In {\em Combinatorics (Proc. Fifth Hungarian Combinatorial Coll.,
  1976., Keszthely)}, page 485–494. Bolyai --- North-Holland, 1978.

\bibitem{PSPACE}
Robert~A. Hearn and Erik~D. Demaine.
\newblock {PSPACE}-completeness of sliding-block puzzles and other problems
  through the nondeterministic constraint logic model of computation.
\newblock {\em Theoretical Computer Science}, 343(1):72--96, 2005.
\newblock \href {https://doi.org/10.1016/j.tcs.2005.05.008}
  {\path{doi:10.1016/j.tcs.2005.05.008}}.

\bibitem{PSPACE-book}
Robert~A Hearn and Erik~D Demaine.
\newblock {\em Games, Puzzles, and Computation}.
\newblock AK Peters/CRC Press, 2009.

\bibitem{ItoZN06}
Takehiro Ito, Xiao Zhou, and Takao Nishizeki.
\newblock Partitioning a graph of bounded tree-width to connected subgraphs of
  almost uniform size.
\newblock {\em J. Discrete Algorithms}, 4(1):142--154, 2006.
\newblock \href {https://doi.org/10.1016/j.jda.2005.01.005}
  {\path{doi:10.1016/j.jda.2005.01.005}}.

\bibitem{JerrumVV86}
Mark Jerrum, Leslie~G. Valiant, and Vijay~V. Vazirani.
\newblock Random generation of combinatorial structures from a uniform
  distribution.
\newblock {\em Theor. Comput. Sci.}, 43:169--188, 1986.
\newblock \href {https://doi.org/10.1016/0304-3975(86)90174-X}
  {\path{doi:10.1016/0304-3975(86)90174-X}}.

\bibitem{LariRPS16}
Isabella Lari, Federica Ricca, Justo Puerto, and Andrea Scozzari.
\newblock Partitioning a graph into connected components with fixed centers and
  optimizing cost-based objective functions or equipartition criteria.
\newblock {\em Networks}, 67(1):69--81, 2016.
\newblock \href {https://doi.org/10.1002/net.21661}
  {\path{doi:10.1002/net.21661}}.

\bibitem{LevinPeres17}
David~A. Levin and Yuval Peres.
\newblock {\em Markov Chains and Mixing Times}.
\newblock AMS, Providence, RI, 2nd edition, 2017.

\bibitem{Lovasz77}
L\'aszl\'o Lov\'asz.
\newblock A homology theory for spanning tress of a graph.
\newblock {\em Acta Mathematica Academiae Scientiarum Hungarica},
  30(3):241--251, 1977.
\newblock \href {https://doi.org/10.1007/BF01896190}
  {\path{doi:10.1007/BF01896190}}.

\bibitem{NDS19}
Lorenzo Najt, Daryl~R. DeFord, and Justin Solomon.
\newblock Complexity and geometry of sampling connected graph partitions.
\newblock {\em Preprint}, 2019.
\newblock \href {http://arxiv.org/abs/1908.08881} {\path{arXiv:1908.08881}}.

\bibitem{SoltanYZ20}
Saleh Soltan, Mihalis Yannakakis, and Gil Zussman.
\newblock Doubly balanced connected graph partitioning.
\newblock {\em {ACM} Trans. Algorithms}, 16(2):20:1--20:24, 2020.
\newblock \href {https://doi.org/10.1145/3381419} {\path{doi:10.1145/3381419}}.

\bibitem{SuzukiTN90}
Hitoshi Suzuki, Naomi Takahashi, and Takao Nishizeki.
\newblock A linear algorithm for bipartition of biconnected graphs.
\newblock {\em Inf. Process. Lett.}, 33(5):227--231, 1990.
\newblock \href {https://doi.org/10.1016/0020-0190(90)90189-5}
  {\path{doi:10.1016/0020-0190(90)90189-5}}.

\bibitem{tutte1956theorem}
William~T Tutte.
\newblock A theorem on planar graphs.
\newblock {\em Transactions of the American Mathematical Society},
  82(1):99--116, 1956.

\bibitem{Vig13}
Giovanni Viglietta.
\newblock Partial searchlight scheduling is strongly pspace-complete.
\newblock In {\em Proc. of the 25th Canadian Conference on Computational
  Geometry (CCCG)}. Waterloo, ON, 2013.

\bibitem{WadaK93}
Koichi Wada and Kimio Kawaguchi.
\newblock Efficient algorithms for tripartitioning triconnected graphs and
  3-edge-connected graphs.
\newblock In {\em Proc. 19th Workshop on Graph-Theoretic Concepts in Computer
  Science ({WG})}, volume 790 of {\em LNCS}, pages 132--143. Springer, 1993.
\newblock \href {https://doi.org/10.1007/3-540-57899-4\_47}
  {\path{doi:10.1007/3-540-57899-4\_47}}.

\bibitem{WuZW16}
Di~Wu, Zhao Zhang, and Weili Wu.
\newblock Approximation algorithm for the balanced 2-connected k-partition
  problem.
\newblock {\em Theor. Comput. Sci.}, 609:627--638, 2016.
\newblock \href {https://doi.org/10.1016/j.tcs.2015.02.001}
  {\path{doi:10.1016/j.tcs.2015.02.001}}.

\end{thebibliography}

\ifthenelse{\boolean{lncs}}{
\appendix
\magicappendix
}{}

\end{document}